\documentclass[11pt]{article}

\usepackage[text={6.8in,8.5in},centering]{geometry}
\usepackage{times}
\usepackage{amsfonts}
\usepackage{amsmath}
\usepackage{amssymb}
\usepackage{amsthm}
\usepackage{graphicx}
\usepackage{subfig} 
\usepackage{url}
\usepackage{hyperref}

\usepackage{color}

\graphicspath{{Fig/}{SlidesFigs/}} 

\usepackage{cancel}

\newtheorem{lemma}{Lemma}[section]
\newtheorem{theorem}[lemma]{Theorem}

\newcommand{\ms}[1]{\ensuremath{\mathsf{#1}}}
\newcommand{\bra}[1]{\ensuremath{\langle#1|}}
\newcommand{\ket}[1]{\ensuremath{|#1\rangle}}
\newcommand{\argmax}{\operatornamewithlimits{arg\ max}}

\renewcommand{\bar}{\overline}

\newcommand{\ver}{{\ms{V}}}
\newcommand{\edge}{{\ms{E}}}

\newcommand{\nbr}{{\ms{nbr}}}
\newcommand{\lens}{{\ms{lens}}}

\newcommand{\gap}{{\ms{min gap}}}
\newcommand{\mingap}{{\ms{min gap}}}

\newcommand{\energy}{{\mathcal{E}}}
\newcommand{\ham}{{\mathcal{H}}}

\newcommand{\wmis}{\ms{mis}}
\newcommand{\mis}{\ms{mis}}
\newcommand{\GS}{\ms{GS}}
\newcommand{\FS}{\ms{FS}}
\newcommand{\mdef}{\stackrel{\mathrm{def}}{=}}

\newcommand{\oy}{{\mathcal{Y}}}

\def\final{1} 
\ifnum\final=1  
\newcommand{\vnote}[1]{[{\small Vicky: \bf #1}]\marginpar{*}}
\newcommand{\sidecomment}[1]{\marginpar{\tiny #1}}
\else 
\newcommand{\vnote}[1]{}
\newcommand{\sidecomment}[1]{}
\fi  

\begin{document}
\title{The Effects of the Problem Hamiltonian Parameters on the Minimum
  Spectral Gap in Adiabatic Quantum Optimization}
\author{
Vicky Choi\\
Gladiolus Veritatis Consulting Co.}

\maketitle       
\begin{abstract}
We study the relation between the
Ising problem Hamiltonian parameters and the minimum spectral gap
(min-gap) of the system
Hamiltonian in the Ising-based quantum annealer.
The main argument we use in this paper to assess the performance of
a QA algorithm is the presence or absence of an {\em anti-crossing}
during quantum evolution. 
For this purpose, we introduce a new parametrization definition of the
anti-crossing. 
Using  the Maximum-weighted Independent Set (MIS) problem in which
there are  flexible parameters (energy penalties $J$
between pairs of edges) in an Ising  formulation as the model problem,
we construct examples to show that by changing the value  of $J$, we can change the quantum
evolution from one that has an anti-crossing (that results in an
exponential small min-gap) to one that does not have, or the other way
around,  and thus drastically change (increase or decrease) the min-gap.
However, we also show that by changing the value of $J$ alone, one can
not avoid the anti-crossing. 
We recall a polynomial reduction from an Ising problem to an MIS
problem to show that the flexibility of changing parameters without
changing the problem to be solved can be applied to any Ising problem.  
As an example, we show that by such a reduction alone, it is possible
to remove the anti-crossing and thus increase the min-gap.
Our anti-crossing definition is necessarily scaling invariant as scaling the problem Hamiltonian does not change the nature
(i.e. presence or absence) of an anti-crossing. As a side note, we
show exactly how the min-gap is scaled if we scale the problem
Hamiltonian by a constant factor. 
\end{abstract}

\section{Introduction}
In this paper, we seek to understand the quantum evolution of the
adiabatic quantum optimization algorithm (see e.g. \cite{AQC-Review} for a recent
  review) by studying how different 
parameters of the Ising problem Hamiltonian affect the minimum spectral
gap ({\em min-gap}) of
the system Hamiltonian in the Ising-based quantum annealer (QA).
We consider the standard quantum annealing protocol with the following
system Hamiltonian:
\begin{align}
\label{eq:Ham1}
  \ham(s) = (1-s) \ham_{\ms{driver}} + s\ham_{\ms{Ising}}
\end{align}
where 
$\ham_{\ms{driver}} =-\sum_{i \in \ver(G)} \sigma^x_i$ is the standard transverse
field Hamiltonian unless otherwise specified, and
$
\ham_{\ms{Ising}} = \sum_{i \in \ver(G)} h_i \sigma^z_i + \sum_{ij \in \edge(G)} J_{ij}
\sigma^z_i \sigma^z_j
$
is an Ising Hamiltonian defined on a graph $G=(\ver(G), \edge(G))$.

We use the Maximum-weighted Independent Set (MIS) as the model problem for our
investigation as it
admits a natural parameter-flexible Ising Hamiltonian formulation,
referred to as the  {\em MIS-Ising
Hamiltonian} (to be defined in Section 3.1). That is, we can change 
parameters ($h, J$) of the MIS-Ising Hamiltonian without changing the problem to
be solved. Moreover, any Ising Hamiltonian can be easily and
efficiently expressed as an MIS-Ising Hamiltonian because any Ising
problem can be so
reduced to an MIS problem (a reduction procedure is described in
Section 3.2).

The main argument we use in this paper to assess the performance of
a QA algorithm is the presence or absence of an {\em anti-crossing}
during quantum evolution. For this purpose, we introduce a new
parametrization definition for the anti-crossing (to be described in
Section 2). The question we ask is then:
what are the factors of a QA algorithm that contribute to the presence or absence of
an anti-crossing that results in an exponentially small min-gap (which
determines the running time of the algorithm)? 


In this paper, 
we propose a preliminary answer for this question, when the initial
ground state is the uniform superposition state (i.e. $\frac{1}{\sqrt{2^n}}
\sum_{i=0}^{2^n-1} \ket{i}$, as in the case of the transverse-field
drive Hamiltonian):
 the presence or absence of an anti-crossing depends on the
relation of the ground state and the first excited state (of the
problem Hamiltonian) with their
low-energy (driver Hamiltonian dependent) neighboring states
(LENS). 
Roughly speaking, with the uniform superposition state as the initial
ground state, if the ground state has {\em more} LENS than the first
excited state does, there is no
anti-crossing and the min-gap is large, but if the first excited state
has {\em more} LENS than the ground state does, there is an anti-crossing
resulting in a small min-gap.
The factors that affect the  LENS include the type of driver
Hamiltonian ($X$-driver: 
$\ham_{X}= - \sum_{i \in \ver(G)}  \sigma_i^x$; or  $XX$-driver:
$\ham_{XX}^{\lambda}= - \sum_{i \in \ver(G)}  \sigma_i^x +  \lambda\sum_{ij \in
  \edge(G_{\ms{driver}})} \sigma_i^x \sigma_j^x$, where $\lambda$ can
be a positive or negative number), the driver
graph $G_{\ms{driver}}$ (in the $XX$-driver case), and the energy spectrum of states (of the
problem Hamiltonian).  In this paper, we consider mainly $X$-driver
Hamiltonian. The case of $XX$-driver Hamiltonian, both stoquastic and
non-stoquastic,  and different driver
graphs are studied in \cite{driver2}.
Based on our LENS observation, we construct examples 
to answer the questions we study, namely, how different parameters can
or can not change the min-gap by the presence or absence of an anti-crossing.

The paper is organized as follows. In Section 2, we introduce a new
parametrization definition for the anti-crossing. 
We revisit in Section 3 the NP-hard Maximum-weighted
Independent Set (MIS) problem.
In Section 3.1 we rederive the parameter-flexible MIS-Ising Hamiltonian.
We recall an efficient reduction to convert an
Ising problem to an MIS problem in Section 3.2. 
As a side example, we apply the Ising-MIS reduction on  the {\em loop gadgets}, the Ising instances that
were constructed by Tameem Albash to have small min-gaps 
to show the evidence of quantum tunneling. We show that the resulting
MIS-Ising Hamiltonians no longer have an
anti-crossing (even without the need to change the parameters) and the
min-gaps are large.
In Section 4, we describe our LENS idea. Then we construct two examples
based on the idea. In Section 4.1, we construct Example 1 to show
that one can change the parameter $J$ (without changing the problem to
be solved) to change the quantum evolution (from the presence of an
anti-crossing to the absence of one,  or vice versa). In Section 4.2, we construct
Example 2 to show that by changing the value of $J$ alone, one can
not avoid the anti-crossing. 
In Section 5, we
show exactly how the min-gap is scaled if we scale the problem
Hamiltonian by a constant factor. Thus there is no need for
renormalization of the parameters in order for the comparison of
different parameter QA algorithms, as we
know exactly how much the contribution is due to the scaling, and how
much is due to the different value of the parameter.
We conclude with a discussion in Section 6.

\section{A Parametrization Definition of An Anti-crossing}
Anti-crossing, also known as avoided level
crossing or level repulsion, is a well-known concept for physicists. 
A parametrized definition of an
anti-crossing is given by Wilkinson in
\cite{Wilkinson1,Wilkinson2}. In this definition, the energy levels in
the neighborhood of an anti-crossing at $s^*$ takes the form of a {\em
  hyperbola}:
\begin{align}
\label{eq:wilkinson}
  E^{\pm}(s) = E(s^*) + B(s -s^*) \pm \frac{1}{2}[\Delta_{min}^2 + A^2(s-s^*)^2]^{1/2}
\end{align}
where $\Delta_{min}$ denote the min-gap size, $A$ and $B$ are
respectively the difference and the mean of the slopes of the asymptotes
of the hyperbola (referred as the asymptotic slopes of the two
curves in the original paper). 

For example, such a definition was
adopted in a recent paper \cite{Pechukas-gas2018} to study
the effect of noise on the quantum system. 
The derivation of the formula is based on the idea that in the
neighborhood of $s^*$ the behavior of these two energy levels can be
described by degenerate perturbation theory: the energy levels are
eigenvalues of the $2 \times 2$ matrix 
$$ H=
  \begin{bmatrix}
    e^+ & h\\
h & e^-
  \end{bmatrix}
.
$$

However, for our purpose (to identify an anti-crossing based on the
numerical diagonization of the Hamiltonian), we introduce the following
parametrization. This parametrization is also based on the same idea
that  in the
neighborhood of $s^*$ the behavior of these two energy levels can be
described by degenerate perturbation theory. But instead of the energy
values alone, we make use of the two strongly mixed eigenstates to identify the avoided crossing point.
It is likely that the
similar but probably non-quantitative idea has been described somewhere in the
literature. 

Let $\ket{E_k(s)}$ ($E_k(s)$  respectively) be the instantaneous
eigenstate (energy respectively) of the system Hamiltonian $\ham(s)$
in Eq.(\ref{eq:Ham1}) at time $s$, i.e., 
$$\ham(s) \ket{E_k(s)} = E_k(s) \ket{E_k(s)}$$ for $k=0,1, \ldots .$
We order and represent the states so that the energies are strictly increasing as
the index increases:
$$
E_0(s) < E_1(s) < E_2(s) < \ldots
$$
That is, if there are some degenerate states, we represent the
corresponding state as a
superposition of the degenerate states. 
We will focus on the lowest two instantaneous  eigenstates, namely, 
$\ket{E_0(s)}$,  the instantaneous ground state, and
$\ket{E_1(s)}$, the instantaneous first excited state.
In particular, $\mingap = \min_{s \in [0,1)} E_1(s) - E_0(s) =
E_1(s^*) - E_0(s^*)$, where $s^*$ is the min-gap position. 
We are interested in whether there is an anti-crossing occurrence at $s^*$.

When $s=1$, the states$\ket{E_k(1)}$ and energies $E_k(1)$ are of the problem (final)
Hamiltonian. Again, we mainly focus on the lowest two levels for $k=0$
and $k=1$. For convenience, we denote
$\ket{\GS} \mdef \ket{E_0(1)}$, the ground state (encode solution,
possibly degenerate)  of the problem
Hamiltonian, and $\ket{\FS} \mdef \ket{E_1(1)}$, the superposition of
the possibly degenerate first excited states. 
We will express the instantaneous eigenstates ($\ket{E_0(s)},
\ket{E_1(s)}$)  in terms of the
eigenstates ($\ket{E_k(1)}$) of the final Hamiltonian.

Define 
\begin{align}
  \label{eq:6}
  a_k(s) = |\bra{E_k(1)}{E_0(s)}\rangle|^2\\
  b_k(s) = |\bra{E_k(1)}{E_1(s)}\rangle|^2
\end{align}
for $k=0,1$. That is, $a_0(s) = |\bra{\GS}{E_0(s)}\rangle|^2$ is the
weight (or overlap) of the solution state with the instantaneous ground state at
time $s$. Similarly, $a_1(s)  = |\bra{\FS}{E_0(s)}\rangle|^2$ is the
weight (or
overlap) of the first excited state (which possibly corresponds to the local
minima of the problem)  with the instantaneous ground state at
time $s$. At $s=1$, we have $a_0(1)=b_1(1)=1$
and $a_1(1)=b_0(1)=0$.
The evolution of  $a_k(s)$ and $b_k(s)$ will play important roles to
help us understand the working of the QA algorithm. 
In particular, we will define
the anti-crossing based on the four quantities $a_k(s)$ and
$b_k(s)$ for $k=0,1$. 

 We now introduce the definition for an anti-crossing, based on two
 parameters: $\gamma \ge 0$ and $\epsilon \ge 0$.
\paragraph{Definition.}
For $\gamma \ge 0$, $\epsilon \ge 0$ we say there is an
{\em $(\gamma, \epsilon)$-Anti-crossing}
 if  there exists a $\delta>0$ such that 
 \begin{enumerate}
 \item For $s \in [s^*-\delta, s^*+\delta]$,
\begin{align}
  \label{eq:1}
  \ket{E_0(s)} = \sqrt{a_0(s)} \ket{\GS} + \sqrt{a_1(s)}\ket{\FS}\\
\ket{E_1(s)} = \sqrt{b_0(s)} \ket{\GS} - \sqrt{b_1(s)}\ket{\FS}
\end{align}
where $a_0(s)+a_1(s) \in [1-\gamma, 1]$, $b_0(s)+b_1(s) \in  [1-\gamma, 1]$. 
Within the time interval $[s^*-\delta, s^*+\delta]$, both
$\ket{E_0(s)}$ and $\ket{E_1(s)}$
are mainly composed of $\ket{\GS}$ and $\ket{\FS}$. That is,
  all other states (eigenstates of the problem Hamiltonian) are
  negligible (which sums up to at most $\gamma \ge 0$).

\item At the avoided crossing point $s=s^*$, 
$a_0, a_1,b_0,b_2 \in [1/2-\epsilon, 1/2+\epsilon]$, for a small
$\epsilon>0$.
That is, $\ket{E_0(s^*)} \approx 1/\sqrt{2} (\ket{\GS}
  +\ket{\FS})$ and $\ket{E_1(s^*)} \approx 1/\sqrt{2} (\ket{\GS}
  -\ket{\FS})$.

\item 
Within the time interval $[s^*-\delta, s^*+\delta]$, $a_0(s)$
increases from $\leq \gamma$ to $\geq (1-\gamma)$, while
  $a_1(s)$ decreases from $\geq(1-\gamma)$ to $ \leq \gamma$. 
The reverse is
  true for $b_0(s), b_1(s)$.
\end{enumerate}

\begin{figure}[h]
  \centering
$$
  \begin{array}[h]{ccc}
  \includegraphics[width=0.44\textwidth]{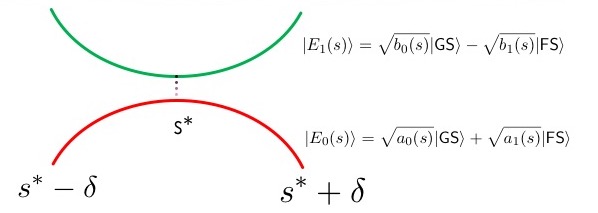} &
\includegraphics[width=0.18\textwidth]{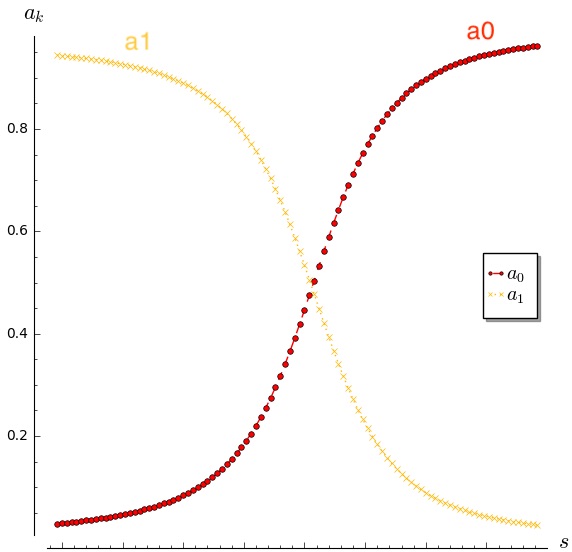} &
\includegraphics[width=0.18\textwidth]{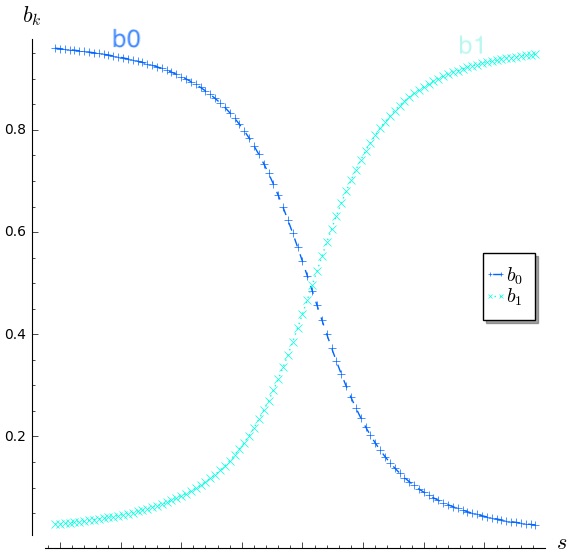} \\
 (a) & (b) & (c)  
  \end{array}
$$
  \caption{(a) An $(\gamma,\epsilon)$-Anti-crossing. At the avoided crossing point $s=s^*$, 
$a_0, a_1,b_0,b_2 \in [1/2-\epsilon, 1/2+\epsilon]$, for a small
$\epsilon>0$.
That is, $\ket{E_0(s^*)} \approx 1/\sqrt{2} (\ket{\GS}
  +\ket{\FS})$ and $\ket{E_1(s^*)} \approx 1/\sqrt{2} (\ket{\GS}
  -\ket{\FS})$. There exists a $\delta>0$, such that  fo $s \in [s^*-\delta, s^*+\delta]$:
$a_0(s)+a_1(s) \in [1-\gamma, 1]$, $b_0(s)+b_1(s) \in  [1-\gamma, 1]$.
(b) Within the time interval $[s^*-\delta, s^*+\delta]$, $a_0(s)$  (in red) 
increases, while
  $a_1(s)$ (in orange)  decreases. 
(c) The reverse is
  true for $b_0(s)$ (in blue), and $b_1(s)$ (in cyan).
}
  \label{fig:anti-crossing}
\end{figure}

A figure depicting an $(\gamma, \epsilon)$-Anti-crossing is shown in
Figure \ref{fig:anti-crossing}. 

\paragraph{Remarks.}
\begin{itemize}
\item 
At the avoided level crossing point $s^*$, the effective
  Hamiltonian  mimics the perturbed degenerate Hamiltonian:
$$
  \begin{bmatrix}
    1 & -\lambda\\
-\lambda & 1
  \end{bmatrix}
$$
with $\ket{\GS} = 1/\sqrt{2}(\ket{0} + \ket{1})$ with eigenenergy
$1-\lambda$, and $\ket{\FS} = 1/\sqrt{2}(\ket{0} - \ket{1})$ with
eigenenergy $1+\lambda$.
\end{itemize}

\paragraph{Parameters and Weak Anti-crossing.}
The parameter $\epsilon$ is the tolerance
we allow at the avoided crossing point. This parameter needs to be
strict and thus we require $\epsilon$ to be very small (e.g. $\le 0.001$).
The parameter $\gamma$ corresponds to the allowed negligible
states. Our numerical examples show that it is possible that the
condition (2) is  satisfied when $\epsilon \approx 0$, and  $a_0(s)+a_1(s)$
is almost $1$ for $s \in  [s^*, s^*+\delta]$, but $a_0(s)+a_1(s)$ is
much less than $1-\gamma$ for $s \in  [s^*-\delta, s^*]$. 
Our numerical examples also show that
when such a situation occurs, the min-gap does not occur at the exact
position of $s^*$. For this reason, we relax the definition and
refer such a case as a {\em weak anti-crossing} with one additional parameter,
$\gamma'$, where $\gamma'>\gamma$ is the relaxed parameter
such that $a_0(s)+a_1(s)>1 - \gamma'$ for $s \in  [s^*-\delta, s^*]$. 
Notice that we maintain the strict requirement at the avoided crossing
point. 
For the weak anti-crossing, before the avoided-crossing point, there
are non-negligible states other than $\ket{\GS}$ and $\ket{\FS}$. But at
the avoided crossing and after avoided crossing, the states are mainly composed
of $\ket{\GS}$ and $\ket{\FS}$.

\paragraph{Anti-crossing and Scaling.}
An anti-crossing is necessarily scaling invariant, as scaling the
problem Hamiltonian should not change the nature (i.e. the
presence or absence) of an anti-crossing. 
We discuss more on this in
Section \ref{sec:scaling}. 

\paragraph{Anti-crossing vs Quantum Tunneling.}
Quantum tunneling is closely related to the anti-crossing. Indeed, one
feature that characterizes the presence of tunneling is by ``a sharp
change in the ground state of the adiabatic evolution at the
degeneracy point ''\cite{tunneling-MAL}. This sharp change is
quantified by the expectation value of the Hamming weight operator
$\langle HW \rangle$
introduced there, which is readily given by $a_1(s) ||\FS-\GS||$, where $||\FS-\GS||$ is the
Hamming distance between $\FS$ and $\GS$, in our anti-crossing
definition. Note that quantum tunneling is more generally described in
terms of a double-well semiclassical potential, see more discussion in \cite{tunneling-MAL}.

\paragraph{Anti-crossing vs Perturbative Crossing.}
Our definition of anti-crossing is more general than the
perturbative crossing in \cite{Amin-Choi,Dickson-Amin}. 
More specifically, 
the  perturbative crossing is defined and  limited to the location near the
end of evolution when the perturbation theory is applied to the problem
Hamiltonian as the unperturbed Hamiltonian;
 whereas our
anti-crossing can happen anywhere (the perturbation theory is still
applicable but the unperturbed Hamiltonian is no longer the pure
problem Hamiltonian).
A perturbative crossing is necessarily an anti-crossing, while an
anti-crossing defined here is not necessarily a perturbative crossing.

\paragraph{Anti-crossing and Min-gap Size.} The min-gap size is
expected to be exponentially small in $O(b^k)$ where $k=
||\FS-\GS||$ for some $0<b<1$. A min-gap estimation formula for the
perturbative crossing was given in \cite{Amin-Choi}. It is desirable
to rigorously derive a bound for the min-gap based on our more general 
anti-crossing definition. 

\paragraph{Generalized Anti-crossing.} One can also generalize the
defintion of the anti-crossing by 
replacing the first excited state $\ket{\FS}$ by a superposition of some
neighboring low lying states. Such an example can be found in \cite{driver2}.

\section{Maximum-Weight Independent Set (MIS) Problem}


The Maximum-Weight Independent Set (MIS)
problem (optimization version) is defined as:

\smallskip
\hspace*{0cm}{\bf Input:} An undirected graph $G (=(\ver(G),\edge(G)))$, where each vertex $i \in \ver(G) = \{1, \ldots, n \}$ is weighted by a
positive rational number $w_i$

\hspace*{0cm}{\bf Output:} A subset $S \subseteq \ver(G)$ such that
$S$ is independent (i.e., for each $i,j \in S$, $i\neq j$, $ij
\not \in \edge(G)$) and the total
{\em weight} of $S$ ($=\sum_{i \in S}
w_i$) is maximized. 
Denote the optimal set by $\wmis(G)$.
\smallskip

We recall a 
quadratic binary optimization formulation (QUBO) of the problem.
More details can be found in \cite{minor1}.
\begin{theorem}[Theorem 5.1 in \cite{minor1}]
If $\lambda_{ij} \ge \min\{w_i,w_j\}$ for all $ij \in \edge(G)$, then the maximum
  value of
  \begin{equation}
\oy(x_1,\ldots, x_n) = \sum_{i \in \ver(G)}w_i x_i - \sum_{ij \in \edge(G)}
  \lambda_{ij}x_ix_j
\label{eq:Y}
  \end{equation}
is the total weight of the MIS. 
In particular if $\lambda_{ij} > \min\{w_i,w_j\}$ for all
      $ij \in \edge(G)$, then $\wmis(G) = \{i \in \ver(G) : x^*_i = 1\}$,
where $(x^*_1, \ldots, x^*_n) = \argmax_{(x_1, \ldots, x_n) \in \{0,1\}^n}
\oy(x_1, \ldots, x_n)$.
\label{thm:mis}
\end{theorem}

Here the function $\oy$ is called the pseudo-boolean function for MIS,
where the boolean variable $x_i \in \{0,1\}$, for $i=1,\ldots, n$.  The proof is quite
intuitive in the way that one can think of $\lambda_{ij}$ as the {\em energy
  penalty} when there is an edge $ij \in \edge(G)$.
In this formulation, we only require $\lambda_{ij} > \min\{w_i,w_j\}$, and thus there is freedom in
choosing this parameter.

\subsection{MIS-Ising Hamiltonian}
By changing the variables ($x_i=\frac{1+s_i}{2}$ where $x_i \in
\{0,1\}, s_i \in \{-1,1\})$, it is easy to show that MIS is equivalent
to minimizing the following function, known as the {\em Ising energy function}:
\begin{eqnarray}
  \energy(s_1, \ldots, s_n) &=& \sum_{i \in \ver(G)} h_i s_i + \sum_{ij \in \edge(G)} J_{ij}s_is_j,
\end{eqnarray}
which is the 
eigenfunction of the following
 {\em Ising Hamiltonian}:
\begin{equation}
\ham_{\ms{Ising}} = \sum_{i \in \ver(G)} h_i \sigma^z_i + \sum_{ij \in \edge(G)} J_{ij}
\sigma^z_i \sigma^z_j
\label{eq:Ising}
\end{equation}
where $h_i = \sum_{j \in \nbr(i)}
  \lambda_{ij} - 2w_i$, (conversely $w_i = 1/2(\sum_{j \in \nbr(i)}
  J_{ij} - h_i)$), $J_{ij}=\lambda_{ij}$, $\nbr(i) =\{j: ij \in \edge(G)\}$,
for $i \in \ver(G)$.

Therefore, different $\lambda_{ij}$ in $\oy$  will correspond to different
$h_i, J_{ij}$ in $\ham$ (notice that $h_i$ is expressed in terms of $J_{ij}$). For convenience, we will refer to a Hamiltonian
in such a form as an {\em MIS-Ising} Hamiltonian.
A natural
question to ask is: is larger $J_{ij}$ ($=\lambda_{ij}$) better (in terms of the 
min-gap)? or is smaller $J_{ij}$
better? 
Intuitively,
larger $J$ penalizes the dependent sets. However, small $J$ allows 
low-energy dependent sets, which can be good if they are neighbors 
to the ground state as we investigate and explain in the sections below.

\subsection{Reduction from the Ising problem to MIS: Ising $\Longrightarrow$ MIS}
In this section, we recall a polynomial reduction to convert an
Ising problem to an MIS problem described in \cite{BorosHammer}. 
Thus, any Ising Hamiltonian can always be reduced to a
parameter-flexible MIS-Ising Hamiltonian. 
The
reduction is quite straightforward. The size of the reduced problem
graph is at most $O(n+m)$ of the original graph where $n$ is the number of vertices
and $m$ is the number of the edges.

The reduction consists of the following 3 steps.

{\bf Step 1}: We change the variables from the spin $s_i \in \{-1,+1\}$ to
the boolean variable $x_i \in \{0,1\}$, that is, from the Ising energy
to a pseudo-boolean function (QUBO). 

{\bf Step 2}: We represent the QUBO in a
{\em posiform}.
The binary variable $x_i$ and its complement $\bar{x_i}=1-x_i$ are
called together literals. Let $L = \{x_1,\bar{x_1}, \ldots, x_n,
\bar{x_n}\}$ denote the set of literals. 
A posiform of a pseudo-boolean function is a polynomial expression in
terms of all literals such that the coefficients are all positive:
$$\phi(x_1,\ldots,x_n) = \sum_{T \subset L}a_T \prod_{u \in T} u$$
where $a_T>0$, $T$ is non-empty and also it does not contain $\{u,
\bar{u}\}$ (otherwise the product will be zero). 
\paragraph{Example.} Suppose $\oy(x_1,x_2,x_3,x_4) = -5.5x_1 -3x_2-3x_3 + 4
x_1x_2 + 4x_1x_3 + 2x_2x_4 + 2 x_3x_4$ is the pseudo-boolean function.
One possible posiform of $\oy$ can be $\phi(x_1,
x_2,x_3,x_4) = -11.5 +5.5\bar{x_1} +3\bar{x_2}+3\bar{x_3} + 4
x_1x_2 + 4x_1x_3 + 2x_2x_4 + 2 x_3x_4$, by replacing $-x_i$ with
$(\bar{x_i} -1)$.

{\bf Step 3}: Given a posiform, we then associate to it a weighted
graph $G_{\phi}$,  called its {\em conflict graph}. Vertices of
$G_{\phi}$ correspond to the nontrivial terms of $\phi$, i.e.
${\mathcal{T}}=\{T \subset L : T\neq \emptyset, a_T>0 \}.$  To a vertex
$T \in {\mathcal{T}}$ we associate $a_T$ as its
weight. Two terms $T,T'$ are in conflict if there is a literal $u \in
T$ for which $\bar{u} \in T'$. The edges of $G_{\phi}$ correspond to the
conflicting pairs of terms. 
Since there are $O(n+m)$ terms in the posiform, where $n,m$ are the
number of vertices and edges of the original graph, there are $O(n+m)$
vertices in the conflict graph.

It was shown in Theorem 3 in \cite{BorosHammer}  that $\max_{(x_1,
  \ldots, x_n)\in \{0,1\}^n}
  \phi(x_1, \ldots, x_n) =
a_{\emptyset} + mis(G_{\phi})$.
Therefore by reducing the original Ising problem to an MIS problem, we
can express it with a parameter-flexible MIS-Ising Hamiltonian
as shown in Section 3.1.

\subsubsection*{Example: Reduction for the Loop Gadgets}
As an example, we apply the reduction to  the {\em loop gadgets}, the Ising instances that
were constructed by Tameem Albash,
to have small perturbative crossing min-gaps 
to show the evidence of quantum tunneling. We show that the resulting
MIS-Ising Hamiltonians, even without the need to change the parameters, no longer have an
anti-crossing (to show the
evidence of the quantum tunneling) and the min-gaps are large.
The loop-gadget  is shown in  Figure~\ref{fig:LG}. In the loop, all couplings
are ferromagnetic.
All local fields are zero except the two ends with $R-1$ and $-R$. All
coupling magnitudes are equal with $-R$, except a
pair with $-R/2$, whose sum is equal to the negative local field. 
By construction, the instance has a unique ground state,
$\ket{0\ldots0}$, and 2-fold degenerate first excited state:
$\ket{1\ldots11}$ and $\ket{1\ldots10}$.
The instances are shown to have exponential small min-gaps by a numerical
fitting (for size up to 20) to $\sim exp(-0.593n)$, as shown in  Figure~\ref{fig:LG}. 

\begin{figure}[h]
  \centering
$$
  \begin{array}{cc}
 \includegraphics[width=0.5\textwidth]{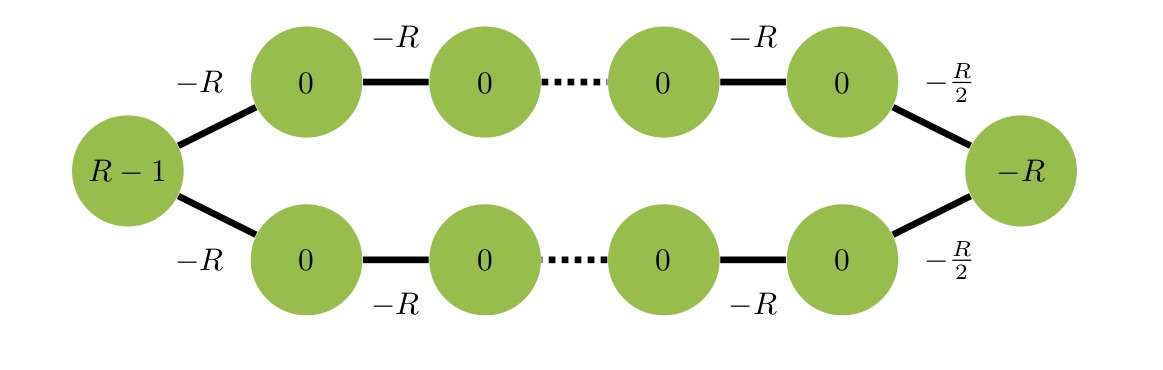} &
\includegraphics[width=0.3\textwidth]{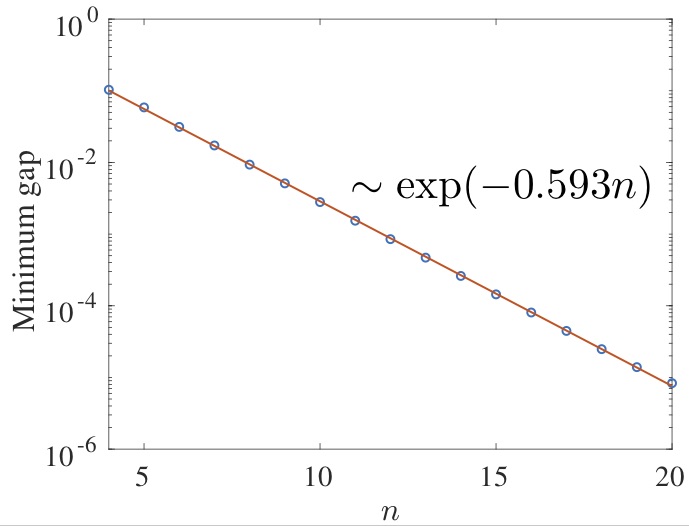}\\
(a) & (b)
  \end{array}
$$ 
 \caption{[Credit to Tameem Albash.](a) The loop gadget. In the loop, all couplings
are ferromagnetic, with the value indicated ($-R$ or $-R/2$).
All local fields are zero except the two ends with $R-1$ and $-R$. 
By construction, the instance has a unique ground state,
$\ket{00\ldots0}$, and 2-fold degenerate first excited state:
$\ket{1\ldots11}$ and $\ket{1\ldots10}$. 
(b) The min-gaps data is fit to  $\sim exp(-0.593n)$. (The loop
Hamiltonian used here is normalized by $R$ ($R=4$ in this example) so that the largest magnitude of any term in the Hamiltonian is 1.)
 }
  \label{fig:LG}
\end{figure}

The  Ising Hamiltonian for the loop gadget with $n=4$ is: 
\begin{align}
  \label{eq:2}
  \ham_{\ms{loop}} = (R-1) \sigma_1^z -R \sigma_4^z -R
  \sigma_1^z\sigma_2^z -R
  \sigma_1^z\sigma_3^z -R/2
  \sigma_2^z\sigma_4^z -R/2
  \sigma_3^z\sigma_4^z 
\end{align}
for some $R \ge 4$. 
The corresponding QUBO function is
\begin{align}
  \label{eq:4}
  \oy_{\ms{loop}} = -(3R-1)/2 x_1  - 3R/4 x_2 - 3R/4 x_3 + R x_1 x_2 +
  R x_1x_3 + R/2 x_2x_4 + R/2 x_3 x_4
\end{align}
One posiform can be:
\begin{align}
  \label{eq:5}
  \phi_{\ms{loop}} = \ms{constant} + (3R-1)/2 \bar{x_1}  + 3R/4 \bar{x_2} + 3R/4 \bar{x_3} + R x_1 x_2 +
  R x_1x_3 + R/2 x_2x_4 + R/2 x_3 x_4
\end{align}
\begin{figure}[h]
  \centering
$$
  \begin{array}[h]{ll}
\includegraphics[width=0.33\textwidth]{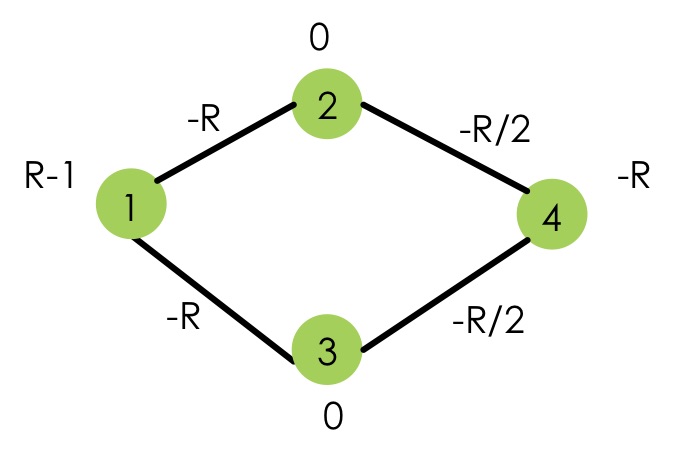} &
 \includegraphics[width=0.33\textwidth]{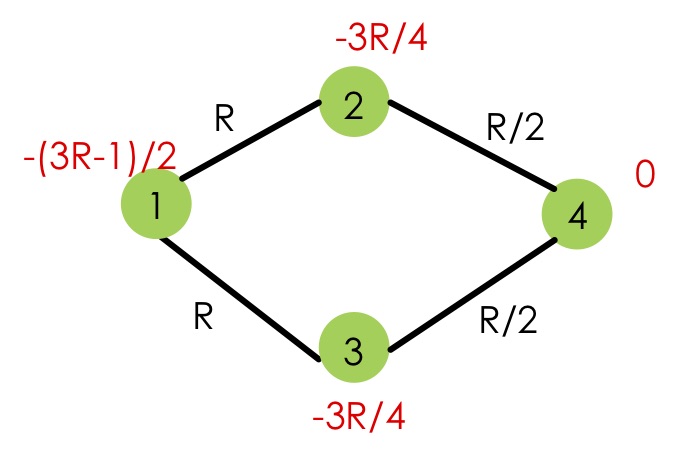} \\
(a) & (b)\\
 \includegraphics[width=0.4\textwidth]{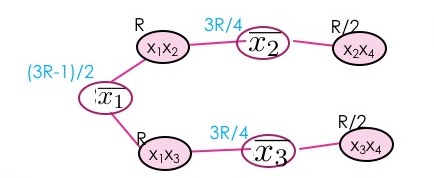} &\\
(c) &
 \end{array} 
$$
 \caption{
(a) The loop-gadget for $n=4$, the corresponding Ising Hamiltonian  in Eq.(\ref{eq:2}).
(b) The corresponding QUBO in Eq.(\ref{eq:4}).
(c)
The corresponding conflict graph of $\phi$ in
    Eq.~(\ref{eq:5}). The MIS of the graph comprises the four vertices in
    pink, with total weight $3R$. That is, the MIS is
    $x_1x_2=x_1x_3=x_2x_3=x_2x_4=1$ implying $x_1=x_2=x_3=x_4=1$,
    the corresponding ground state. The
    three white vertices $\{\bar{x_1}, \bar{x_2}, \bar{x_3}\}$ with weight
    $3R-1/2$ is a local maximum, corresponding to the degenerate first
    excited state with $x_1=x_2=x_3=0, x_4=0/1$.}
  \label{fig:conflict}
\end{figure}
The conflict graph associated with $\phi$ is shown in
Figure~\ref{fig:conflict}. 
Notice that in the corresponding MIS formulation,
there is no longer a 
degenerate first excited state. In general, an $n$ loop is converted to a chain of $2n-1$ vertices,
without the first excited degeneracy.
There is no longer
an anti-crossing, and the min-gap is large. 




\section{Low-Energy Neighboring Eigenstates (LENS)} 
In this section, we describe our LENS idea.
First, we recall the basics of the QA algorithm. 
The algorithm relies
on the fact that if the system is initialized in the ground state of
the driver Hamiltonian, 
the state $\ket{E_0(s)}$ evolves according to the
Schrodinger equation, will
remain as the instantaneous ground state of the Hamiltonian $\ham(s)$ if
the evolution is slow enough, according to the Adiabatic Theorem.
The idea of the QA algorithm is then, that if we encode the
problem into the final  Hamiltonian, at the end of the quantum
evolution we will get the solution (ground state). 
From the algorithmic point of view, 
how does the algorithm or the quantum evolution actually work? In particular, how does
$a_k(s) (= |\bra{E_k(1)}{E_0(s)}\rangle|^2)$, k=0,1, the weight (or the overlap)
of the solution state ($\ket{\GS}=\ket{E_0(1)}$) or the first excited state ($\ket{\FS}=\ket{E_1(1)}$) with the instantaneous 
ground state $\ket{E_0(s)}$ evolve?
How does the driver Hamiltonian affect 
the evolution? 
At $s=1$, we know $a_0(s) =1,a_1(s)=0$. 
Supposing the initial
ground state is the uniform (with positive amplitudes) superposition
state, i.e. $\ket{E_0(0)}=\frac{1}{\sqrt{2^n}} \sum_{i=0}^{2^n-1}
\ket{i}$, we have $a_0(0) = \frac{m_0}{\sqrt{2^n}}, a_1(0) =
\frac{m_1}{\sqrt{2^n}}$, where $m_0$ and $m_1$ are the number of
degenerate states in $\ket{\GS}$ and $\ket{\FS}$ respectively.
Supposing  there is no anti-crossing, what are the factors of the
algorithm that will affect the change in the
value of $a_0(s)$ and $a_1(s)$? 
Our observation is that the amount of change in $a_k(s)$ depends on the
corresponding state $\ket{E_k(1)}$'s {\em neighboring states} and their energy
values (w.r.t the problem Hamiltonian).  

Next, we define what we mean by neighboring eigenstates. The
neighborhood of an eigenstate depends on the driver Hamiltonian. 
As an example, consider the $X$-driver 
$\ham_{X}= - \sum_{i \in \ver(G)}  \sigma_i^x$ as the driver
Hamiltonian.
Recall that $\sigma_i^x$ flips the $i$th qubit, that is, $\sigma_i^x
\ket{x_1x_2\ldots x_i \ldots x_n} = \ket{x_1x_2\ldots \bar{x_i} \ldots
  x_n}$.
We define $\nbr_{\ham_\ms{driver}}(\ket{\phi})=\{\ket{\psi}:
\ket{\psi}=\ms{Op}_i(\ket{\phi})\}$ where $\ham_\ms{driver}=\sum_i
\ms{Op}_i$.
In other words, we have $\bra{\phi} \ham_\ms{driver} \ket{\psi} =1$ iff
$\ket{\psi} \in \nbr_{\ham_\ms{driver}}(\ket{\phi})$.
Thus, by this definition, 
$\nbr_{\ham_\ms{X}}$ consists of the single-bit flip neighborhood of
the state. 
For example, 
$\nbr_{\ham_\ms{X}}(\ket{10101}) = \{\ket{1010\underline{0}},
\ket{101\underline{1}1},
\ket{10\underline{0}01}, \ket{1\underline{1}101},
\ket{\underline{0}0101}
\}$.

Under the assumption that  the initial
ground state is the uniform (with positive amplitudes) superposition
state (i.e. $\ket{E_0(0)}=\frac{1}{\sqrt{2^n}} \sum_{i=0}^{2^n-1}
\ket{i}$), before an anti-crossing,
the ground state $\ket{\GS}=\ket{E_0(1)}$ gets {\em positive contributions}
from its higher energy neighboring states and $a_0(s) (= |\bra{\GS}{E_0(s)}\rangle|^2)$ increases;
similarly, the first excited state $\ket{\FS}=\ket{E_1(1)}$ gets {\em positive contributions}
from its higher energy neighboring states and $a_1(s) (= |\bra{\FS}{E_0(s)}\rangle|^2)$ increases.
The contribution is proportional to
the energy difference between the two states; the contribution is greater if the energy
of the neighboring state is lower (closer to the state's energy value). 
Therefore, it is the low-energy neighboring states (LENS) that will
affect the change in $a_k(s)$. We denote the restricting neighborhood
by $\lens$, that is, $\lens(\ket{\phi}) =
\nbr(\ket{\phi})|_{\ms{low-energy}}$.
Notice that LENS depends on the problem Hamiltonian and the driver
Hamiltonian, but not the evolution path (with the assumption that
there is not yet an anti-crossing occurrence). 
Initially, $a_0(0) = \frac{m_0}{\sqrt{2^n}}, a_1(0) =
\frac{m_1}{\sqrt{2^n}}$, where $m_0$ and $m_1$ are the number of
degenerate states in $\ket{\GS}$ and $\ket{\FS}$ respectively.
As $s$ increases,  if $\ket{\GS}$
has more LENS than $\ket{\FS}$ does,  $a_0(s) $ will increase faster and become dominant
during the evolution. In this case, there is no anti-crossing and the
min-gap is large. However, if $\ket{\FS}$ has more LENS than
$\ket{\GS}$ does, $a_1(s)$ will increase
faster and
become dominant before an anti-crossing. In this case,
there is a small min-gap.

We illustrate our LENS idea in detail in the following
examples. While our examples are constructed based
on this idea,  the results from these examples also serve to
reinforce the correctness of the idea.

\subsection{Example 1: Changing $J$ changes the min-gap}
We construct this example to show that by changing the parameters of the Ising problem Hamiltonian
one can change the quantum evolution (from one that has an
anti-crossing to one that does not, or vice-versa) and thus drastically change the min-gap. 
The input  is a  chain of $5$ weighted
vertices, shown in Figure \ref{fig:path5}. 
\begin{figure}[h]
  \centering
  \includegraphics[width=0.6\textwidth]{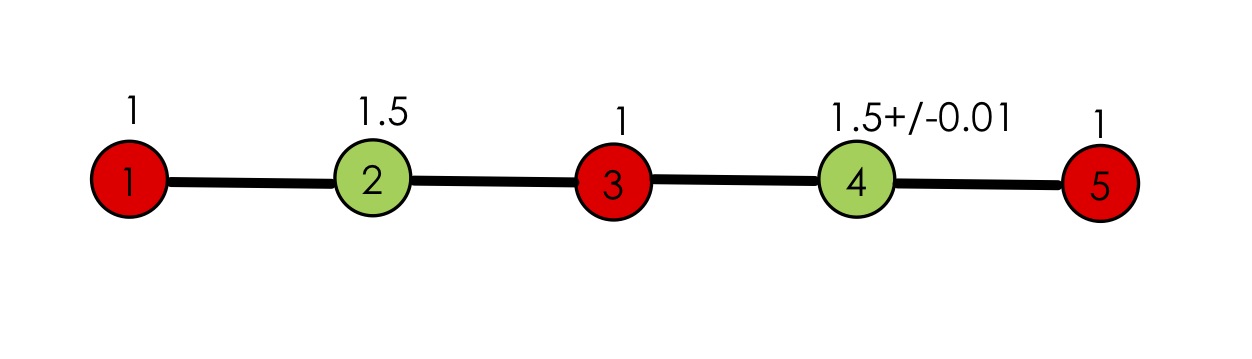}
  \caption{A chain  of 5 weighted vertices. The weight of each vertex is
    indicated above the vertex, where vertex 4 has two possible
    values, either $w_4=1.5+0.01=1.51$ or $w_4=1.5-1.001=1.49$.
Both $\{1,3,5\}$ (in red)  and $\{2,4\}$ (in green) are
    the maximal independent set --- weight$(\{1,3,5\})=3$;
    weight$(\{2,4\})=3.01$ when $w_4=1.51$ or $2.99$ when $w_4=1.49$.
Thus, 
$\mis=\{1,3,5\}$ or 
$\ket{\GS}=\ket{10101}$ for
$w_4=1.49$;  while $\mis=\{2,4\}$ or $\ket{\GS}=\ket{01010}$ for $w_4=1.51$. 
They are complementary to each other. 
}
  \label{fig:path5}
\end{figure}
There are two instances of
this graph, one with $w_4=1.49$, another with $w_4=1.51$. 
Their
solutions are complementary to each other by construction, namely,
$\mis=\{1,3,5\}$ or 
$\ket{\GS}=\ket{10101}$ for
$w_4=1.49$;  while $\mis=\{2,4\}$ or $\ket{\GS}=\ket{01010}$ for $w_4=1.51$. 
Furthermore, these two instances are constructed so that  the ground
state of one instance is the first excited state of the other, and
vice versa.
Thus, the opposite result (the presence vs absence of an
anti-crossing) is expected in the two instances.

In particular, 
the two instances have the opposite results for the min-gap when
increasing  the energy penalty $J$. 
More specifically, for the instance with $w_4=1.49$, when $J=1.52$,
there is an anti-crossing resulting in a small min-gap, see
Figure~\ref{fig:w149J152} for the detailed explanation; but as $J$
increases, the anti-crossing disappears and the min-gap is large,
e.g. see
Figure~\ref{fig:w149J4} for $J=4$.
The opposite results are observed for the instance with $w_4=1.51$,
see Figure~\ref{fig:w151} for the comparison of the results for
$J=1.52$ (absence of the anti-crossing) vs $J=4$ (presence of the
anti-crossing).
The results of various $J$ $(1.52,4,10,100)$\footnote{We purposely
  left out the results for $J=2$ as an exercise for the curious reader.} for both instances with
$w_4=1.49$ vs $w_4=1.51$ are compared side-by-side in Figure~\ref{fig:diff4} (only
the evolutions of $a_k$ are shown). For
    $w_4=1.49$ instance (left in Figure ~\ref{fig:diff4}), increasing $J$  increases min-gap
    (anti-crossing disappears); while for
    $w_4=1.51$ instance (right in Figure ~\ref{fig:diff4}), increasing $J$ decreases min-gap
    (anti-crossing appears).
\begin{figure}[t]
  \centering
  \includegraphics[width=0.6\textwidth]{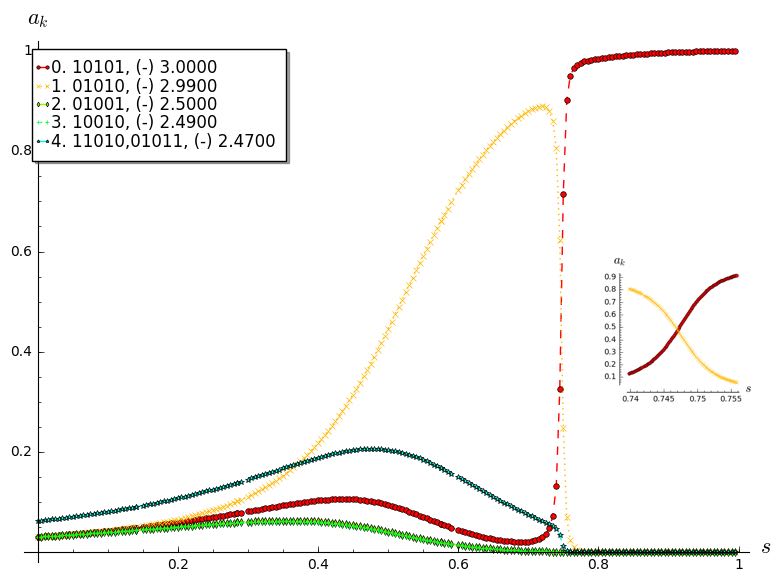} \\
(a)  Evolution of $a_k(s)$, $k=0,1,\ldots, 4$\\
  \includegraphics[width=0.6\textwidth]{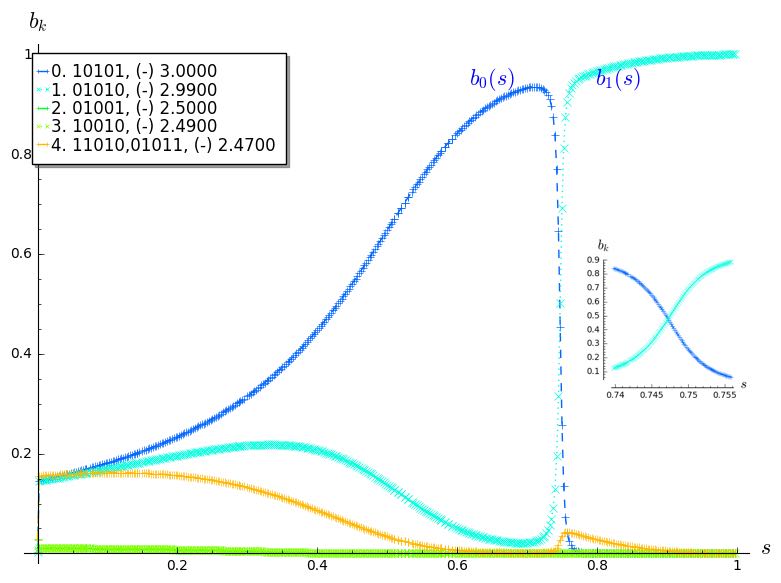} \\
(b) Evolution of $b_k(s)$, $k=0,1,\ldots, 4$
  \caption{For the instance with $w_4=1.49$ and $J=1.52$. (a) Evolution of $a_k(s)$, where $a_k(s) =|\bra{E_k(1)}{E_0(s)}\rangle|^2$ is the overlap  of the
    $\ket{E_k(1)}$ with the ground state wavefunction
      $\ket{E_0(s)}$. The x-axis is the time $s$, and the y-axis is
      $a_k$. The lowest 5 levels ($0 \le k \le 4$) are shown. 
      The one in red is $a_0(s)$ and the one in orange is $a_1(s)$. 
In this instance,  $s^*=0.7479, \gap=0.0018$. 
the inset shows the occurrence of a
$(0.15,0)$-Anti-crossing (with $\delta=0.008$)
at $s^*$. In particular, within the zoom interval,  $a_0(s)$ increases
from $\sim 0.15$ to $\sim 0.85$; while $a_1(s)$ decreases from  $\sim
0.85$ to $\sim 0.15$. At $s^*$, $a_0(s) =a_1(s)=0.5$. 
(b) Evolution of $b_k(s)$, where $b_k(s) = |\bra{E_k(1)}{E_1(s)}\rangle|^2$ is the overlap  of the
    $\ket{E_k(1)}$ with the first excited state wavefunction
      $\ket{E_1(s)}$. 
      The one in blue is $b_0(s)$ and the one in cyan is $b_1(s)$.  In
      the inset, $b_1(s)$ increases
from $\sim 0.15$ to $\sim 0.85$; while $b_0(s)$ decreases from  $\sim
0.85$ to $\sim 0.15$. At $s^*$, $b_0(s) =b_1(s)=0.5$. Together the
above two plots demonstrate the presence of the anti-crossing at
$s^*=0.7479$ where the min-gap occurs. }
  \label{fig:w149J152}
\end{figure}

\begin{figure}[t]
  \centering
  \includegraphics[width=0.6\textwidth]{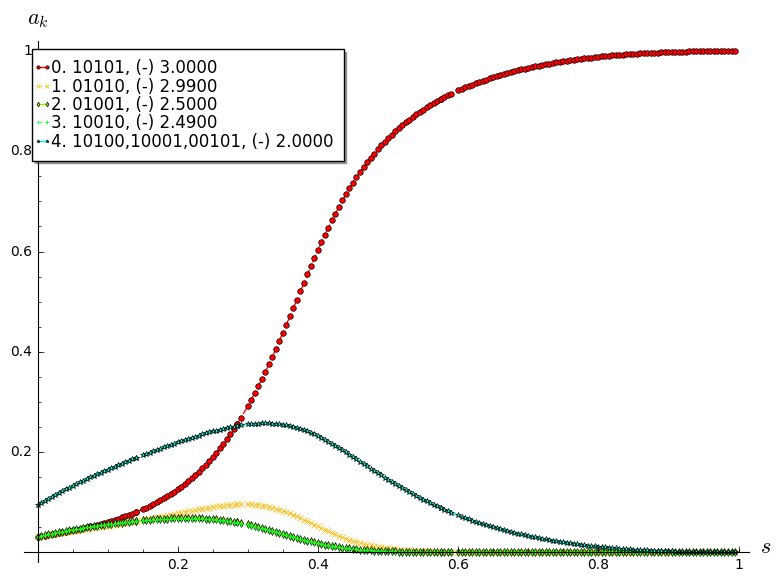} \\
(a)  Evolution of $a_k(s)$\\
  \includegraphics[width=0.6\textwidth]{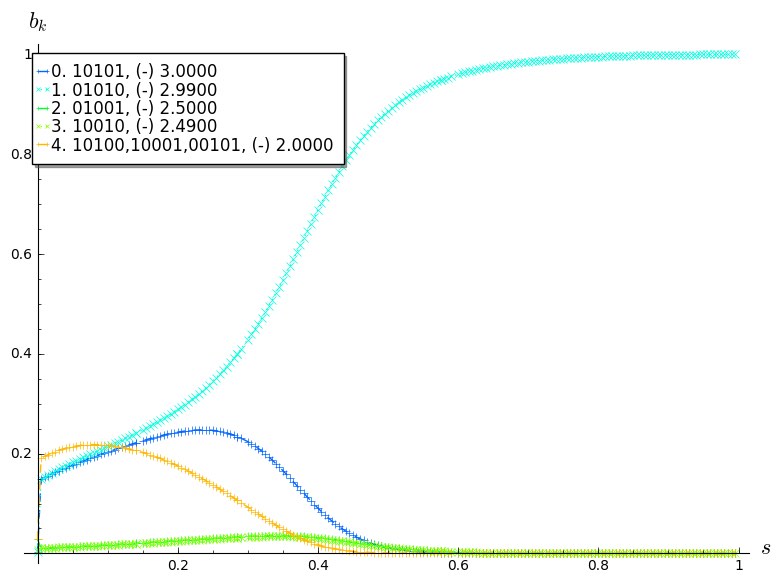} \\
(b) Evolution of $b_k(s)$
  \caption{For the instance with $w_4=1.49$ and $J=4$. (a) Evolution of $a_k(s)$, where $a_k(s) =|\bra{E_k(1)}{E_0(s)}\rangle|^2$ is the overlap  of the
    $\ket{E_k(1)}$ with the ground state wavefunction
      $\ket{E_0(s)}$. 
In this instance,  $s^*=0.935, \gap=0.0387$. There is no anti-crossing
in this case, as one can see $a_0(s)$ increases steadily to $1$ at the
end. 
(b) Evolution of $b_k(s)$, where $b_k(s) = |\bra{E_k(1)}{E_1(s)}\rangle|^2$ is the overlap  of the
    $\ket{E_k(1)}$ with the first excited state wavefunction
      $\ket{E_1(s)}$.}
  \label{fig:w149J4}
\end{figure}

\begin{figure}[t]
  \centering
$$
  \begin{array}[h]{cc}
   \includegraphics[width=0.4\textwidth]{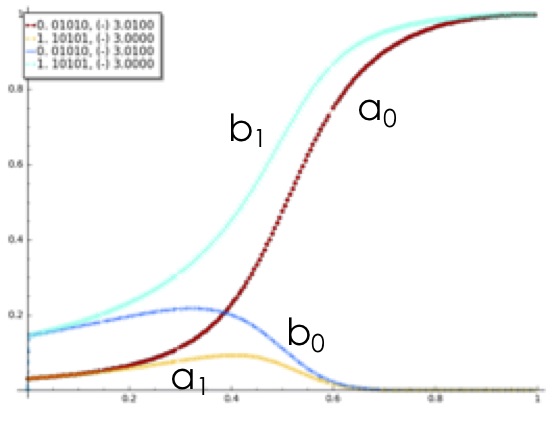} &
  \includegraphics[width=0.4\textwidth]{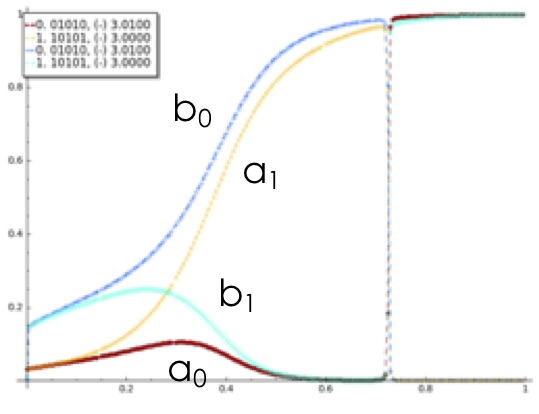} \\
(a) J=1.52 & (b) J=4
  \end{array}
$$ 
 \caption{Evolution of $a_0(s), a_1(s), b_0(s), b_1(s)$ for the
    instance with $w_4=1.51$, where 
$a_k(s) =|\bra{E_k(1)}{E_0(s)}\rangle|^2$ is the overlap  of the
    $\ket{E_k(1)}$ with the ground state wavefunction
      $\ket{E_0(s)}$, ($a_0$ in red, $a_1$ in orange), 
 $b_k(s) = |\bra{E_k(1)}{E_1(s)}\rangle|^2$ is the overlap  of the
    $\ket{E_k(1)}$ with the first excited state wavefunction
      $\ket{E_1(s)}$, ($b_0$ in blue, $b_1$ in cyan).  (a) $J=1.52$: there is no
    anti-crossing. $a_0(s)$ and $b_1(s)$ keep increasing. 
(b) $J=4$: there is an anti-crossing. 
}
  \label{fig:w151}
\end{figure}

\begin{figure}[t]
  \centering
$$
  \begin{array}[c]{l|l|l}
   & w_4=1.49, mis=\{1,3,5\} &w_4=1.51, mis=\{2,4\}\\
\hline\\
J=1.52 & s^*=0.7479, \gap=0.0018 & s^*=0.95, \gap=0.03889\\
& \includegraphics[width=0.36\textwidth]{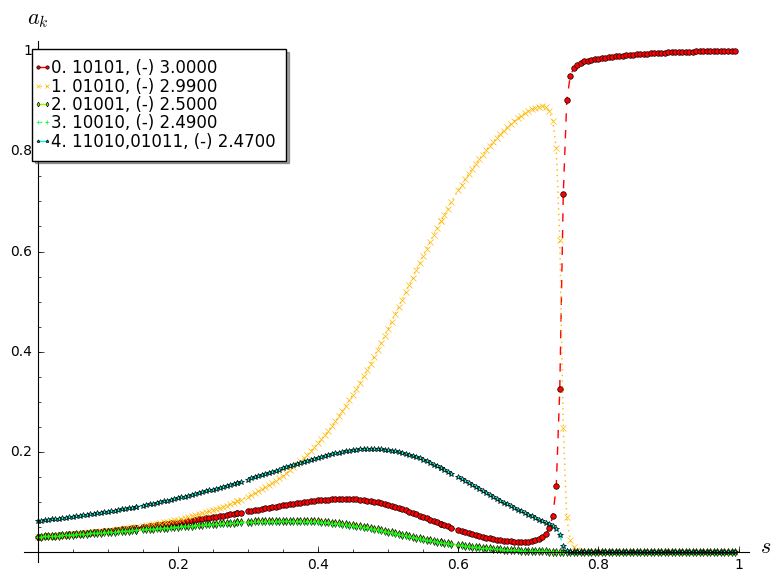} & \includegraphics[width=0.36\textwidth]{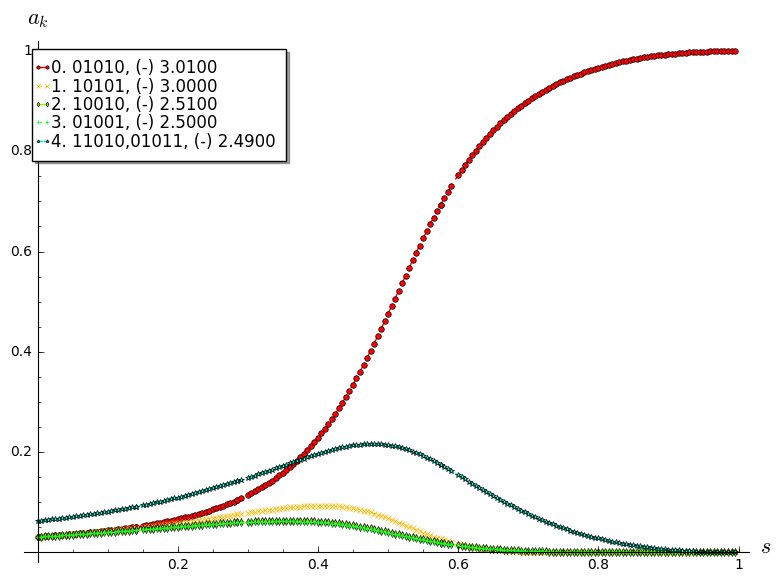}\\
\hline\\
J=4 & s^*=0.94, \gap=0.0387 & s^*=0.7262, \gap=3.8e-4\\
& \includegraphics[width=0.36\textwidth]{ak-w149-J4.jpg} & \includegraphics[width=0.36\textwidth]{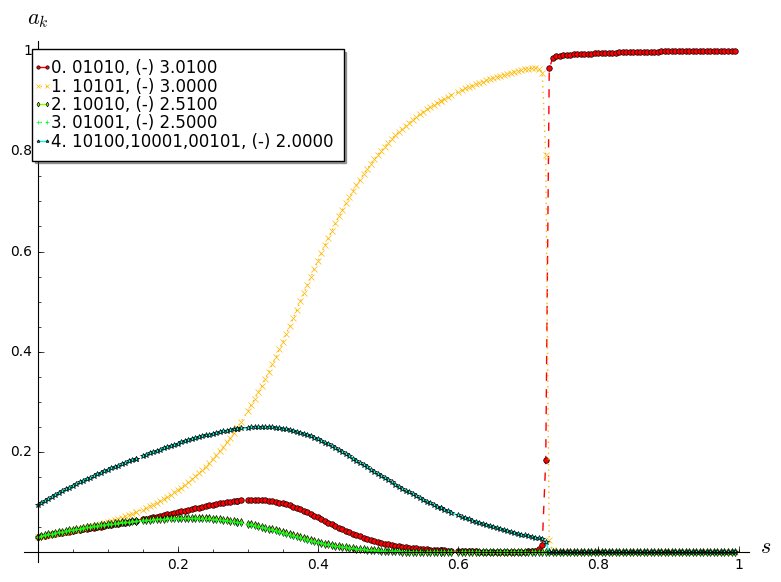}\\
\hline\\
J=10 & s^*=0.95, \gap=0.0389 & s^*=0.7522, \gap=1.2e-4\\
& \includegraphics[width=0.36\textwidth]{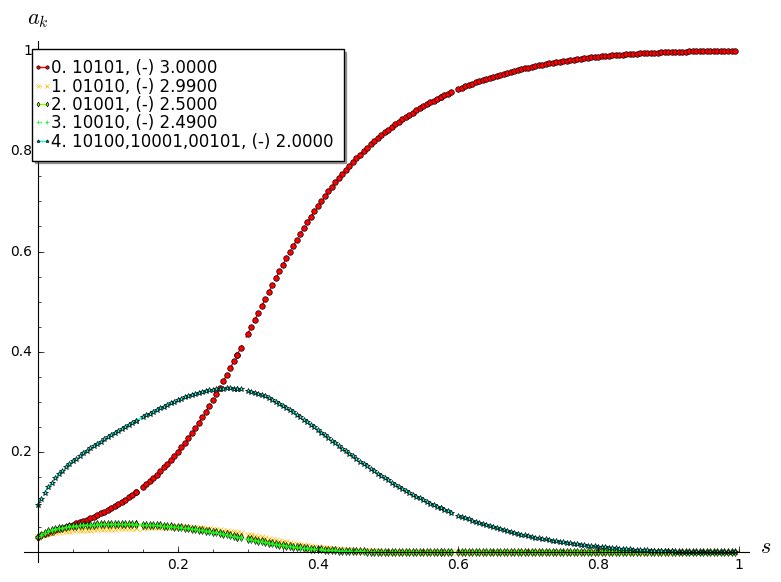}
& \includegraphics[width=0.36\textwidth]{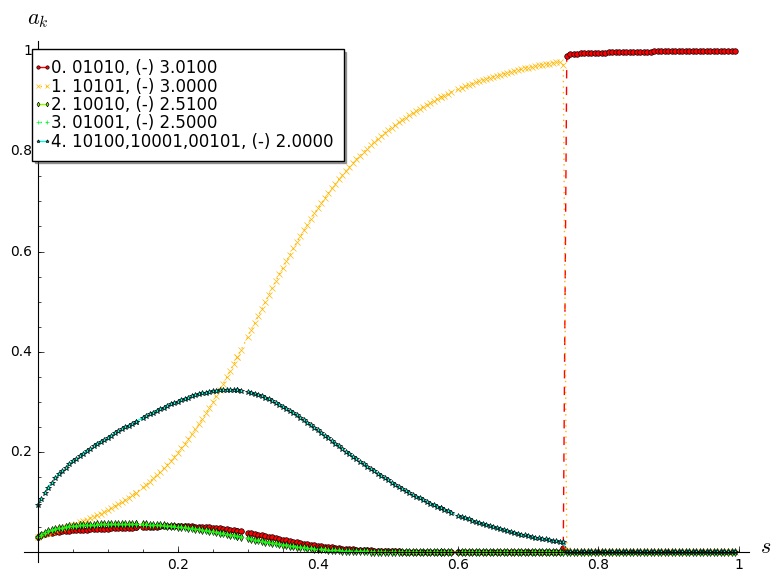}\\
\hline\\
J=100 & s^*=0.95, \gap=0.0391 & s^*=0.76147, \gap=7.136e-5
  \end{array}
$$
  \caption{Evolution of $a_k(s)$  for two different instances
    ($w_4=1.49$ vs $w_4=1.51$) with various $J$.  For
    $w_4=1.49$ instance (left), increasing $J$  increases min-gap
    (anti-crossing disappears); while for
    $w_4=1.51$ instance (right), increasing $J$ decreases min-gap
    (anti-crossing appears). Notice that for the same $J$ the results
    for the two instances with ($w_4=1.49$ vs $w_4=1.51$) are
    opposite (presence vs absence of the anti-crossing).}
  \label{fig:diff4}
\end{figure}

We now apply the LENS idea to explain our results. First, let us try
to understand the result for the instance with $w_4=1.49,
J=1.52$, where there is an anti-crossing as shown in
Figure~\ref{fig:w149J152}. 
In this example, $\ket{\GS}=\ket{10101}$ and $\ket{\FS}=\ket{01010}$.
Among the 5 neighbors of $\ket{\GS}=\ket{10101}$, three of them 
$\{1010\underline{0}, 10\underline{0}01, \underline{0}0101\}$ are
independent sets, of energy value $(-)2$. The other two $\{1\underline{1}101,
101\underline{1}1\}$ are dependent sets of higher energy.
Similarly,  among the 5 neighbors of $\ket{\FS}=\ket{01010}$, there are two
independent sets, and three
dependent sets. The lowest 17 eigenstates of this instance are shown in
Figure~\ref{fig:nbrs-example}, where $\nbr_{\ham_\ms{X}} (\ket{0}) = \{\ket{5}, \ket{11},
\ket{12}\}$ and  $\nbr_{\ham_\ms{X}} (\ket{1}) = \{\ket{4}, \ket{9},\ket{10},
\ket{16}\}$. When restricting to the low-energy, we have
$\lens_{\ham_\ms{X}} (\ket{0}) = \{\ket{5}\}$
and $\lens_{\ham_\ms{X}} (\ket{1}) = \{\ket{4}\}$.
Since $\ket{4}$ is
lower than $\ket{5}$, we say $\ket{\FS}$ has more LENS than
$\ket{\GS}$ does.  In this case, $a_1(s)$ increases faster as $s$ increases
and becomes dominant before the anti-crossing.

However, when $J$ increases to $4$, the ranking of the neighboring
states is changed. The lowest 17 eigenstates of this instance are shown in
Figure~\ref{fig:nbrs-example} (b).
Now we have $\nbr_{\ham_\ms{X}} (\ket{0}) = \{\ket{4}, \ket{15},
\ket{16}\}$ and  $\nbr_{\ham_\ms{X}} (\ket{1}) = \{\ket{5}, \ket{6},\ket{9},
\ket{17}\}$. When restricting
to the low-energy, we have $\lens_{\ham_\ms{X}} (\ket{0}) = \{\ket{4}
\}$ and  $\lens_{\ham_\ms{X}} (\ket{1}) = \{\ket{5},\ket{6}\}$.
Thus, $\ket{\GS}$ has more LENS than $\ket{\FS}$ does, and $a_0(s)$ increases steadily and
there is no anti-crossing, as shown in Figure~\ref{fig:w149J4}.

\begin{figure}[h]
  \centering
$$
  \begin{array}[h]{cc}
    \includegraphics[width=0.4\textwidth]{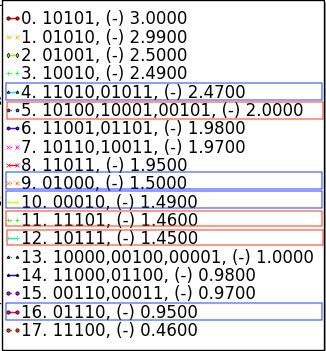} &
  \includegraphics[width=0.4\textwidth]{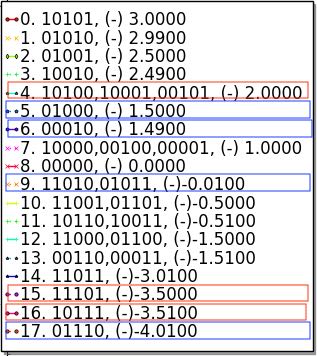}\\
(a) & (b)
\end{array}
$$
  \caption{The lowest 17 eigenstates of the instance with (a) $w_4=1.49$
    and $J=1.52$ and (b) $w_4=1.49$
    and $J=4$. The neighboring states of the ground state $\ket{0}$
    (the first excited state $\ket{1}$, resp.)
    are in red (blue, resp.) rectangles.
For (a), $\nbr_{\ham_\ms{X}} (\ket{0}) = \{\ket{5}, \ket{11},
\ket{12}\}$ and  $\nbr_{\ham_\ms{X}} (\ket{1}) = \{\ket{4}, \ket{9},\ket{10},
\ket{16}\}$.
For (b), $\nbr_{\ham_\ms{X}} (\ket{0}) = \{\ket{4}, \ket{15},
\ket{16}\}$ and  $\nbr_{\ham_\ms{X}} (\ket{1}) = \{\ket{5}, \ket{6},\ket{9},
\ket{17}\}$.
}
  \label{fig:nbrs-example}
\end{figure}

In general, by increasing $J$, we increase the energy of the
dependent set states, e.g.,  $\{\ket{11010},\ket{01011}\}$, which are the neighboring states
to $\ket{01010}$. 
If $J$ is increased such that
$\{\ket{11010},\ket{01011}\}$ becomes the high-energy state (from
being the low-energy state), $\ket{01010}$ will then have less LENS.
If $\ket{01010}$ is the ground state as in
$w_4=1.51$ case, this will reduce the weight of $a_0(s)$ but increase the
weight of $a_1(s)$ and force an anti-crossing. Conversely, if
$\ket{01010}$ is the first excited state as in
$w_4=1.49$ case, this will reduce the weight of $a_1(s)$ and $a_0(s)$ will
keep increasing, and there is no anti-crossing. 
In our above example, for $J=1.52$,
both $\{\ket{11010},\ket{01011}\}$ are low in energy in eigenstate
$\ket{4}$, which is a neighboring state to $\ket{\FS}$. Thus,
$\ket{\FS}$ has more LENS than $\ket{\GS}$ does,
$a_1(s)$ increases  faster as $s$ increases
and becomes dominant before the anti-crossing.
However, for $J=4$, their energy increases to the eigenstate
$\ket{9}$, and $\{\ket{10100}, \ket{10001}, \ket{00101}\}$ ,
which are the neighboring state to $\ket{\GS}$,
becomes $\ket{4}$. 
Thus, $\ket{\GS}$ has more LENS than $\ket{\FS}$, and $a_0(s)$ increases steadily and
there is no anti-crossing. 
As $J$ keeps increasing, it continues to increase the energy value of some dependent sets,
but the low-energy neighborhood structure remains the same, and thus the evolution
remains the same shape with a slight change in the gap size. 

For a fixed $J$, we compare the two different  $w_4$ instances (each row
in Figure \ref{fig:diff4}).  Their results are opposite, one with an
anti-crossing, and one without an anti-crossing. As we mentioned, this
is to be expected. The ground state for one is the first excited state
of the other, and vice versa.
If there is no anti-crossing for one, 
the other will have an anti-crossing
 because the ground state of one instance will become
the  first excited state of the other instance to compete with its complement.


To further verify the
neighborhood part of our LENS idea, we replace the $X$-driver
$\ham_{X}= - \sum_{i \in \ver(G)}  \sigma_i^x$ with a (stoquastic) $XX$-driver
$\ham_{XX}= - \sum_{i \in \ver(G)}  \sigma_i^x - \sum_{ij \in
  \edge(G)} \sigma_i^x \sigma_j^x$ which has the uniform (with
positive amplitudes) superposition state as the  initial ground
state, and thus our LENS idea applies. 
By changing from the $X$-driver to the $XX$-driver, the
 neighborhood will also include the two-bit flip neighbors. 
For example, 
 the
low energy states $\ket{\underline{10}010}$ and $\ket{010\underline{01}}$ become the neighboring
(2-bit flip)
states of $\ket{01010}$. If $\ket{01010}$ is the ground state as in
$w_4=1.51$ case, $\ket{\GS}$ will have more LENS, $a_0(s)$ will
increase faster, and there is no anti-crossing. 
But if $\ket{01010}$ is the first excited state as in the
$w_4=1.49$ case, 
$\ket{\FS}$ will have more LENS, and 
this will reduce the weight of $a_0(s)$ but increase the
weight of $a_1(s)$ and force an anti-crossing.
For $w_4=1.51$, the opposite results when using
$X$-driver vs $XX$-driver are as shown in Figure \ref{fig:XvsXX}.
However, for the same reason, $XX$-driver can also decrease the min-gap
by introducing an anti-crossing while there is no anti-crossing with
$X$-driver. 
See Figure~\ref{fig:XvsXX-w149} for the results for $w_4=1.49$, $J=4, 10$. 
More discussion on the $XX$-driver Hamiltonians (both stoquastic and
non-stoquastic) and the driver graphs can be found in  \cite{driver2}.


\begin{figure}[t]
  \centering
$$
  \begin{array}[c]{l|l|l}
   & \mbox{X-driver: }  \ham_{X}= - \sum_{i \in \ver(G)}  \sigma_i^x &
   \mbox{XX-driver: } \ham_{XX}= - \sum_{i \in \ver(G)}  \sigma_i^x - \sum_{ij \in
  \edge(G)} \sigma_i^x \sigma_j^x\\
\hline\\

J=4 & s^*=0.7262, \gap=3.8e-4 & s^*=0.965, \gap=0.03928\\
& \includegraphics[width=0.36\textwidth]{ak-w151-J4.jpg} & \includegraphics[width=0.36\textwidth]{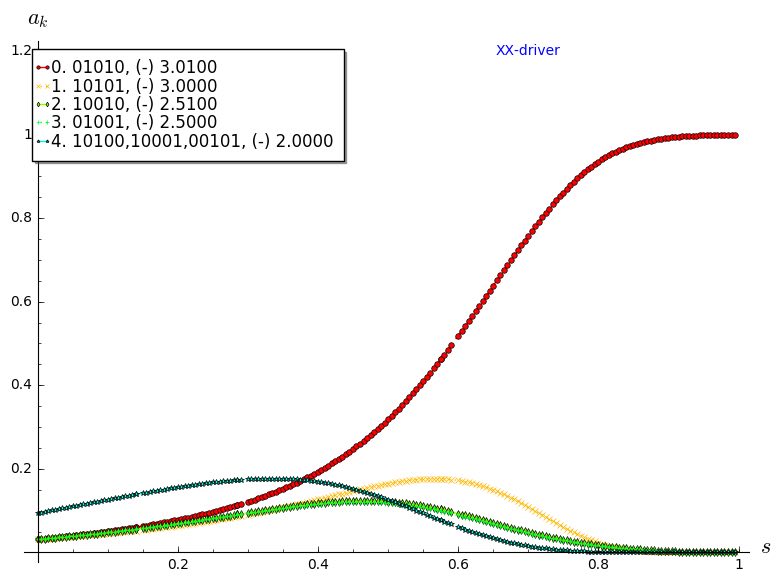}\\
\hline\\
J=10 & s^*=0.7522, \gap=1.2e-4 & s^*=0.965, \gap=0.039322\\
& \includegraphics[width=0.36\textwidth]{ak-w151-J10.jpg}
& \includegraphics[width=0.36\textwidth]{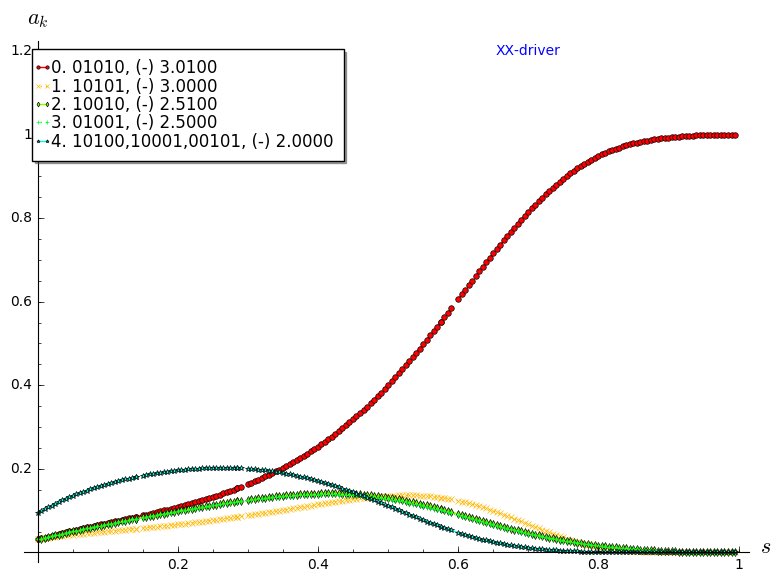}\\
  \end{array}
$$
  \caption{X-driver vs XX-driver for instance $w_4=1.51$. 
The anti-crossing is avoided with the $XX$-driver. This is because
$\ket{2}=\ket{\underline{10}010}$ and $\ket{3}=\ket{010\underline{01}}$ 
become $\lens_{\ms{H_{XX}}}(\ket{\GS})$ where
$\ket{\GS}=\ket{01010}$. Thus, $\ket{\GS}$ has more LENS than
$\ket{\FS}$ does.
}
  \label{fig:XvsXX}
\end{figure}


\begin{figure}[t]
  \centering
$$
  \begin{array}[c]{l|l|l}
& \mbox{X-driver: }  \ham_{X}= - \sum_{i \in \ver(G)}  \sigma_i^x &
   \mbox{XX-driver: } \ham_{XX}= - \sum_{i \in \ver(G)}  \sigma_i^x - \sum_{ij \in
  \edge(G)} \sigma_i^x \sigma_j^x\\
\hline\\

J=4 & s^*=0.94, \gap=0.0387 & s^*=0.82375, \gap=0.016349] \\
& \includegraphics[width=0.36\textwidth]{ak-w149-J4.jpg} & \includegraphics[width=0.36\textwidth]{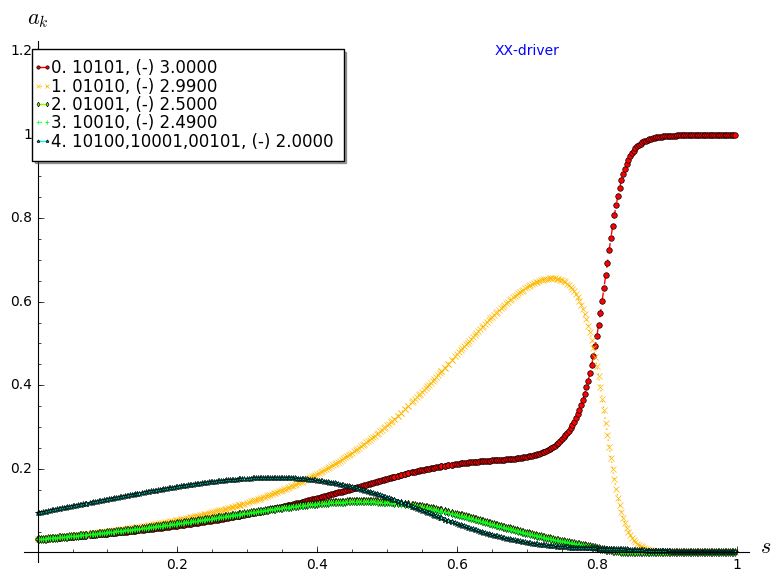}\\
\hline\\
J=10 & s^*=0.95, \gap=0.0389 & s^*=0.81759, \gap=0.013135\\
& \includegraphics[width=0.36\textwidth]{ak-w149-J10.jpg}
& \includegraphics[width=0.36\textwidth]{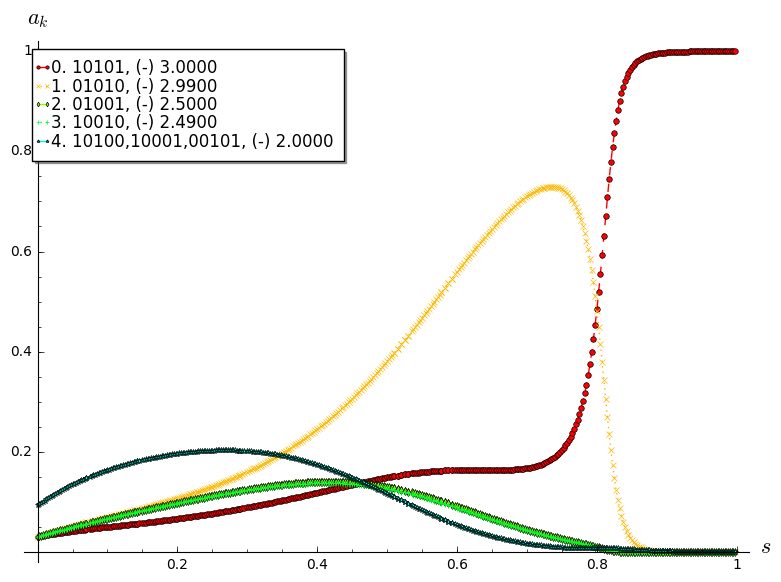}\\
  \end{array}
$$
  \caption{X-driver vs XX-driver for instance $w_4=1.49$. 
The $XX$-driver in this case decrease the min-gap with a weak
anti-crossing. This is because
$\ket{2}=\ket{\underline{10}010}$ and $\ket{3}=\ket{010\underline{01}}$ 
become $\lens_{\ms{H_{XX}}}(\ket{\FS})$ where
$\ket{\FS}=\ket{01010}$.
It is a weak anti-crossing because the min-gap does not happen at $0.5$
anti-crossing point, and also the left boundary does not sum up to
one: there are other non-negligible states besides $\ket{\GS}$ and $\ket{\FS}$.
}
  \label{fig:XvsXX-w149}
\end{figure}

\subsection{Example 2: Anti-crossings remain}
We construct this example to show that it is not always sufficient
that increasing the energy penalty $J$ will remove anti-crossings.
This example is constructed to further verify  our LENS idea.
The input graph is 
a chain of 7 qubits as shown in Figure \ref{fig:path7}. 
The results with different $J$ are shown in Figure \ref{fig:diff4-2}. 
Recall
that increasing $J$ will increase the energy of the dependent set
states. However, this example is so constructed that all the
low-energy states are independent sets, and thus, increasing $J$ only
increases the already high-energy neighboring states, and thus it does
not effectively change LENS. Furthermore,  the states among the
low-energy spectrum are at least 3 bit-flips from either $\ket{\GS}$
or $\ket{\FS}$, therefore changing to the stoquastic $XX$-driver Hamiltonian does not help.
Indeed, in this example, the lowest 6 states remain the same
even if with a larger $J$ value, or using a larger neighborhood.

\begin{figure}[h]
  \centering
  \includegraphics[width=0.6\textwidth]{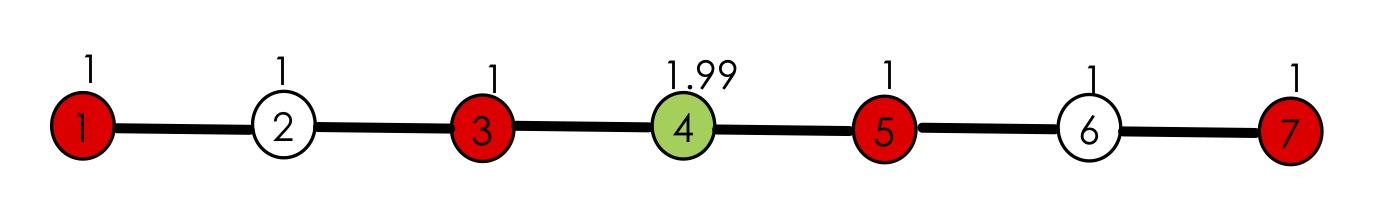}
  \caption{A chain  of 7 weighted vertices. The weight of each vertex
    is 
    indicated above the vertex, where  vertex $w_4=1.99$ and all other
    vertex has weight 1. The mis is $\{1,3,5,7\}$ with wight $4$. But
    there are  a 4-fold degenerate local minima
    $\{146,246,147,247\}$,  with weight $3.99$.}
  \label{fig:path7}
\end{figure}

\begin{figure}[t]
  \centering
$$
  \begin{array}[c]{ll}
J=2, s^*=0.84375, \gap=0.00155 & J=10, s^*=0.70609, \gap=0.00413683\\
\includegraphics[width=0.5\textwidth]{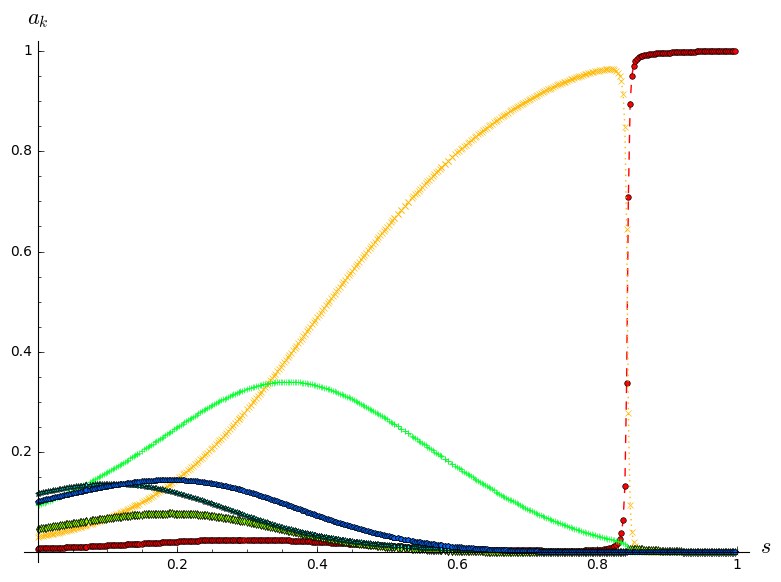} & \includegraphics[width=0.5\textwidth]{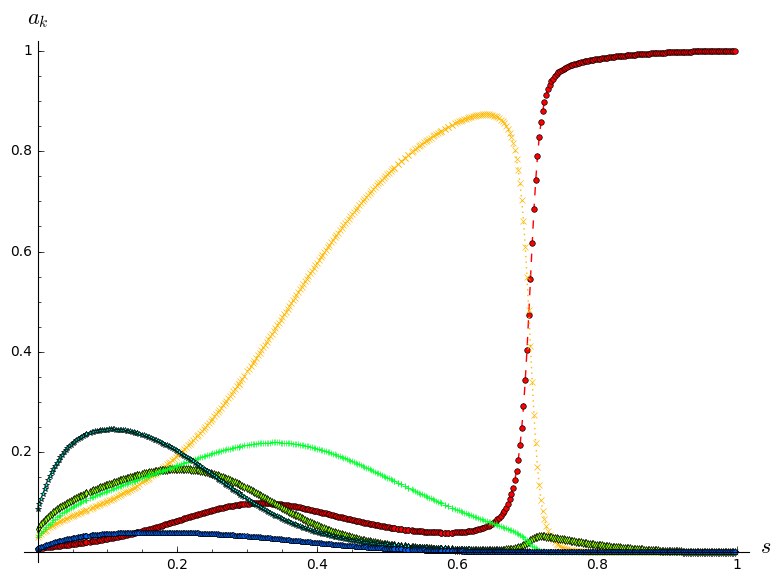} \\
J=100, s^*=0.6456, \gap=0.0074294 & J=1000, s^*=0.6389, \gap=0.0080297\\
  \end{array}
$$
\includegraphics[width=0.9\textwidth]{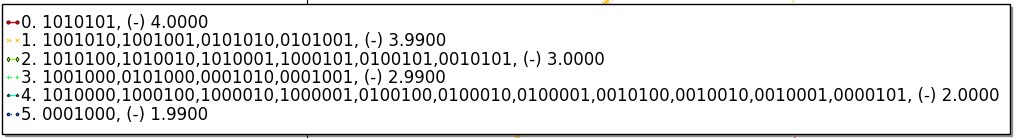} 
  \caption{Different $J$ for the chain-7 instance in Figure~\ref{fig:path7}. The
    lowest 5 states consist of all independent set states and remain
    the same for any $J$. The increase in the
    min-gap as $J$ increases is mainly due to the scaling.}
  \label{fig:diff4-2}
\end{figure}

\paragraph{Remark.}
The arguments
made by Dickson and Amin in \cite{Dickson-Amin} to remove the perturbative
crossings do not apply to the above two examples. Take for instance
the $w_4=1.51$ instance in Example 1. In this example there is
no degeneracy. Increasing $J$ actually introduces the presence of an
anti-crossing and  decreases the min-gap: the
opposite of their argument that it would eliminate the perturbative
crossing and increase the min-gap.  
There are two main differences:
(1) Our anti-crossing is more general than the perturbative crossing
which only applies to the perturbation at the end
of evolution with respect to the problem Hamiltonian.
(2) Our example is for the weighted MIS problem, while the argument
there is for the unweighted MIS problem.

\section{Scaling the Problem Hamiltonian}
\label{sec:scaling}
In this section, we investigate how the minimum
spectral gap  of a scaled
Hamiltonian $H^{\alpha}$ varies according to the parameter $\alpha$.
In particular, we consider the $\alpha$-scaled Hamiltonian:
\begin{align}
  \label{eq:3}
H^{\alpha}(t) =(1-t)H_B + t \alpha \cdot H_P    
\end{align}
where $t \in [0,1]$, $\alpha > 1$. (Remark: We consider $\alpha>1$, and will compare
$H^{\alpha}$ with $H^1$. In the case $\alpha<1$,  one can compare
$H^{\alpha}$ with $(H^{\alpha})^{1/\alpha}$.)
Let $H^{1}=H(s) = (1-s) H_B + s H_P$ where $s\in [0,1]$.

\begin{theorem}
  \begin{enumerate}

\item There is one-one correspondence between $s$ and $t$.  Namely,
  for $s \in [0,1)$, $t(s) =  \frac{s}{\alpha(1-s) +s}$. 
Conversely, for $t \in [0,1)$, $s(t) =\frac{t\alpha}{1+ (\alpha -1)t}$.

\item The eigen-energies are scaled, and the eigenstates are preserved:
\begin{align}
  \label{eq:118}
  E_i^{\alpha}(t) = (1+(\alpha-1)t) E_i(s(t))
\end{align}
where $E_i^{\alpha}(t)$ ($E_i(s(t))$ resp.) are the $i$th eigenvalue
of $H^{\alpha}(t)$ ($H(s(t))$ resp.), for all $i$.
Furthermore, $\ket{E_i^{\alpha}(t)} = \ket{E_i(s(t))}$, for all $i$.
  \item The min-gap is scaled, and the min-gap position is shifted:
  $$\gap(H^\alpha)  = (1+f(\alpha))
  \gap (H^1)$$ where $f(\alpha))  = (\alpha -1) t' >0$, with $t^* \le
  t' \le t(s^*)$, $t^*$ ($s^*$ resp. ) is the position of the minimum gap of
$H^{\alpha}$ ($H^{1}$ resp.), and $t(s) =  \frac{s}{\alpha(1-s) +s}$
for $s \in [0,1)$.
Moreover, $t^* \le t(s^*) < s^*$.
  \end{enumerate}
\end{theorem}

\begin{proof}
  Let $H(s) = (1-s) H_B + s H_P$ where $s\in [0,1]$.
We will show there is one-one correspondence between $s \in [0,1)$ in $H(s)$ and
$t\in [0,1)$ in $H^{\alpha}(t)$.
Consider 
$$
\left\{ 
\begin{array}{l}
\frac{H(s)}{1-s} = H_B + \frac{s}{1-s} H_P\\
  \frac{H^{\alpha}(t)}{1-t}  =H_B + \frac {t \alpha}{1-t} H_P
\end{array}
\right.
$$

Solving $\frac{s}{1-s} = \frac {t \alpha}{1-t}$, we get 
$s(t) := s =\frac{t\alpha}{1+ (\alpha -1)t} \in [0,1)$, and
$t(s) := t = \frac{s}{\alpha(1-s) +s} \in [0,1)$. Notice that $s(t)$
and $t(s)$ are the bijections and inverse of each other. 

Thus, we have 
$$ \frac{H^{\alpha}(t)}{1-t} = \frac{ H(s)}{1-s}$$
where $s=s(t)$ or $t=t(s)$.

Since $\frac{1-t}{1-s}= \frac{1-t}{ 1- \frac{t\alpha}{1+(\alpha-1)t}}
= 1 +(\alpha -1)t$, we have
\begin{align}
  \label{eq:7}
H^{\alpha}(t) = \frac{1-t}{1-s}H(s)=
(1+(\alpha-1)t)
H(s) 
\end{align}

Therefore, we have
\begin{align}
  \label{eq:8}
  E_i^{\alpha}(t) = (1+(\alpha-1)t) E_i(s(t))
\end{align}
where $E_i^{\alpha}(t)$ ($E_i(s(t))$ resp.) are the $i$th eigenvalue
of $H^{\alpha}(t)$ ($H(s(t))$ resp.).
Also, $\ket{E_i^{\alpha}(t)} = \ket{E_i(s(t))}$.
In particular, $\gap(H^\alpha) = (1+(\alpha-1)t^*)
\mbox{gap}_H(s(t^*)) \ge (1+(\alpha-1)t^*) \gap(H)$. 

Similarly, $\gap(H^{\alpha}) \le \mbox{gap}_{H^{\alpha}}(t(s^*)) =
(1+(\alpha -1)t(s^*)) \gap(H)$.


Next we show that $t^* \le t(s^*)< s^*$.

Since $\alpha>1$, $t(s) = \frac{s}{\alpha(1-s) +s} <
\frac{s}{(1-s) +s}=s$, in particular $t(s^*) < s^*$. 

From Eq.~\eqref{eq:8},
$$
\left\{
\begin{array}{ll}
\mbox{gap}_{H^{\alpha}}(t(s^*)) &= (1+(\alpha-1)t(s^*))
\mbox{gap}_H(s^*)\\
\mbox{gap}_{H^{\alpha}}(t^*) &= (1+(\alpha-1)t^*)
\mbox{gap}_H(s(t^*))  
\end{array}
\right.
$$
Since $\gap(H) = \mbox{gap}_H(s^*) \le 
\mbox{gap}_H(s(t^*))$, if $t(s^*) < t^*$, it would imply
$\mbox{gap}_{H^{\alpha}}(t(s^*)) < (1+(\alpha-1)t^*)
\mbox{gap}_H(s^*) = (1+(\alpha-1)t^*)\gap(H) 
<  (1+(\alpha-1)t^*) \mbox{gap}_H(s(t^*))
< \mbox{gap}_{H^{\alpha}}(t^*) =
\gap(H^{\alpha})$ a contradiction.
\end{proof}

\paragraph{Remarks:}
\begin{itemize}
\item The above theorem implies that for $\alpha>1$,  the minimum gap increases as $\alpha$
increases.
\item Notice that $t(s^*)$ and $t^*$ are not
  necessarily the same. For $\alpha>1$, $t^* \le t(s^*)< s^*$.
This implies that by increasing $\alpha$, we also push the position
of the minimum gap towards $0$. This is because when $\alpha$ is large, $t^*
\leq t(s^*) = \frac{s^*}{\alpha(1-s^*) +s^*} \leadsto 0$. This may
also pose a physical limitation of the value of $\alpha$ as too close
to zero may shorten the ``quantum phase''.
\item If we allow infinite energy value, e.g. by setting $\alpha = 1 +
  K/t(s^*)$, for an arbitrarily large $K$,  we have $\gap(H^\alpha) \ge (1+K)
  \gap (H^1)$. That is, if we allow infinite energy
  value (which is unrealistic), we can have the min-gap arbitrarily
  large. 
\end{itemize}


\paragraph{Anti-crossing preserved under scaling.}
If the curvature around the min-gap is sharp(as in the
exponentially small gap case),  $t^* = t(s^*)$, and $s(t^*)=s^*$.
Then,
\begin{align}
  \label{eq:scale-factor}
\gap(H^\alpha)  = (1+(\alpha -1)t(s^*))
  \gap (H^1).
\end{align}
What is more, 
$\ket{E_0^{\alpha}(t^*)} = \ket{E_0(s^*)}$.  Hence
the presence of the anti-crossing will be preserved under our parametrized definition.
To illustrate, we use $w_4=1.51, J=10$ in Example 1. The scaling
or renormalizing factor  is $\alpha=10$. The comparison is shown in Figure \ref{fig:scaling}.
Furthermore, if it is a weak anti-crossing, by scaling, one can not
make it to a (strong) anti-crossing.

\paragraph{Renormalization.} 
When comparing the performance  of
different parameter QA algorithms, there is an argument that one needs to renormalize the problem
Hamiltonian so that the largest parameter is of the same
value. Indeed, scaling/renormalizing the problem Hamiltonian will
change the min-gap according to our above theorem. 
As we show in Figure \ref{fig:scaling}, we know exactly how much the contribution to the min-gap is due to
the scaling (according to Eq~(\ref{eq:scale-factor})), and how much is due to the difference in the parameters. Hence,
there is no such  need for renormalization, when one can use the
parameter value as it is.




\begin{figure}[t]
  \centering
$$
  \begin{array}[h]{cc}
   \includegraphics[width=0.4\textwidth]{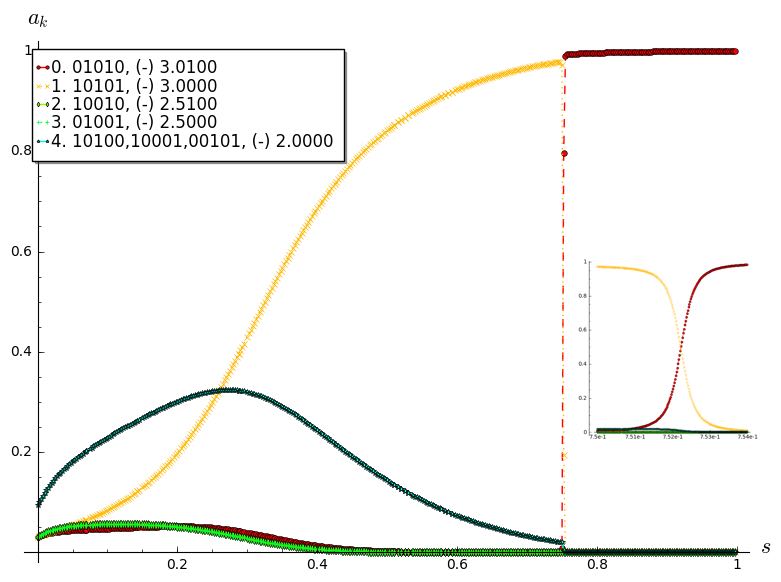} &
  \includegraphics[width=0.4\textwidth]{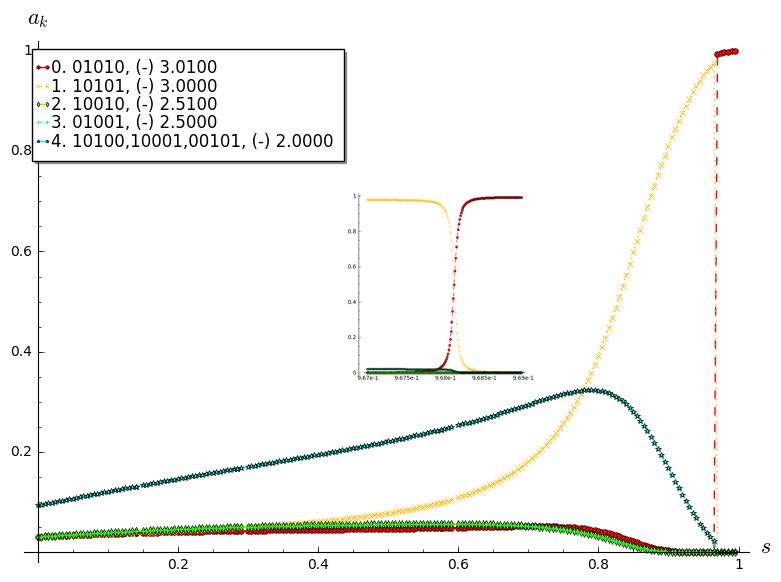} \\
(a) \alpha=10 & (b) \alpha=1 
  \end{array}
$$ 
 \caption{Evolution of $a_k(s)$ for the
    instance with $w_4=1.51$, and $J=10$ of Example 1.
The $(\gamma, \epsilon)$-Anti-crossing is preserved with two different
scaling factors (a) $\alpha=10$ and  (b) $\alpha=1$, where $\epsilon=0$ and
$\gamma \le 0.001$ for both, $\delta=0.02$ for (a), but $\delta=0.002$ for (b).
Moreover, (a) $t^*=0.7522$,
$\mingap=1.2e-4$; and (b) $s^*=0.9681$, $\mingap=1.566e-5$.
One can check that $t^*=t(s^*)$, the min-gaps satisfy $\gap(H^\alpha)  = (1+(\alpha -1)t(s^*))
  \gap (H^1)$, where the factor $(1+(\alpha -1)t(s^*)= 7.7698$, $\alpha=10$.
}
  \label{fig:scaling}
\end{figure}



\section{Summary and Discussion}
In this paper, we show that the problem Hamiltonian parameters can
affect the minimum spectral gap of the adiabatic algorithm. 
The main argument we use  to assess the performance of
a QA algorithm is the presence or absence of an anti-crossing
during quantum evolution. For this purpose, we introduce a new parametrization definition of the
anti-crossing. 
Our definition allows us to characterize or provide a
signature for identifying an anti-crossing based on the numerical
diagonalization of the Hamiltonian.  
This is in contrast
to other numerical studies that
compute the min-gaps for some small number of instances of size up to
20, and then best fit the data with some exponentially small function
(see e.g.
Figure~\ref{fig:LG}). Besides, sometimes the evidence of quantum
tunneling that was quantified by a sharp change in the ground state
expectation of the Hamming weight operator  $\langle
HW \rangle$, such as Figure 6 in \cite{Non-stoquastic}, was presented
for further justification. Anti-crossing is intimately related to the quantum tunneling. Indeed,
the loop-gadgets example was constructed to have small anti-crossing
gaps in order to show the evidence of the quantum tunneling.
As we remarked, our anti-crossing definition readily gives rise to the values of the Hamming weight operator  $\langle
HW \rangle$.
Our anti-crossing definition reflects the known
concept of the anti-crossing (c.f. \cite{Wilkinson1}) and
is more general
than the perturbative crossing in \cite{Amin-Choi,Dickson-Amin}. A
perturbative crossing is necessarily an anti-crossing, but an anti-crossing is not necessarily
a perturbative crossing which is limited to the location near the
end of evolution when the perturbation theory is applied to the problem
Hamiltonian as the unperturbed Hamiltonian.
In \cite{Amin-Choi}, an estimation formula of the min-gap size was given
based on the perturbation theory. We
are investigating how to estimate the min-gap size based on this more general
anti-crossing definition. 

Based on our LENS
observation that the presence or absence of an anti-crossing depends
on the relation of the ground state and the first excited state with
their LENS, we construct two Maximum-weighted Independent Set (MIS) examples to answer the questions we
study. More specifically, we construct Example 1 to show that one can
change the energy penalty parameter $J$ (without changing the problem to be solved) 
to change the quantum evolution (from the presence of an anti-crossing
to the absence, or the other way around). However, we also show in Example 2  that by
changing the value of $J$ alone, 
one can not avoid the anti-crossing. 
The examples we construct are of small sizes for easy
illustration. It is not difficult to construct a larger size of
instances such that the argument still holds (e.g. one can construct
instances which contain our small graph as subgraphs). 
Admittedly, our LENS idea is still premature and  incomplete because the general situation can be much more complex, with many more
anti-crossings. Nevertheless, for what its worth, we believe that
this simple observation is useful  to give some understanding of
the working of the QA algorithm, especially in this early stage of
QA research where rigorous  small experimental tests are needed. 
Moreover, by understanding the causes
of the formation of the anti-crossing, one can also learn to come up
with possible ways
of avoiding such bottlenecks, as we show in our Example 1. 

We construct our examples of the MIS problem because it has a natural
parameter-flexible Ising formulation. However, our results do not lose generality because any Ising problem can be easily and efficiently reduced
to the parameter-flexible MIS-Ising Hamiltonian and thus also allow the  flexibility to change the
parameters without changing the problem to be solved. Based on our
results, one possible advantage of expressing as an MIS-Ising
Hamiltonian is then that one
may {\em adaptively} re-run the program based on the measured
output to improve the performance.
For example, one idea is as follows: assume that the output state
$\ket{\phi}$  is $\ket{\FS}$ (which has more LENS according to the
original setting). We then increase the energy penalty of
$\ket{\phi}$'s neighbors (again we can do so because of the MIS-Ising Hamiltonian). 
This will effectively increase $\ket{\phi}$'s neighboring states (dependent set)
in the next run,
and thus discourage the state from becoming dominant in the early evolution if
$\ket{\phi}$ is not the true ground state. 
In the case of Example 1, it is possible to improve the algorithm
performance. However, as we show in Example 2, this is not always
sufficient to remove the anti-crossing.

We also show exactly how the min-gap is scaled if we scale the problem
Hamiltonian by a constant factor. In particular, one can
increase the min-gap to arbitrarily large if the energy value could be
infinite (which is unphysical).
Furthermore, our result implies that there is no need for
renormalization of the parameters in order for the comparison of
different parameter QA algorithms.
These results are to raise attention to the importance of the dynamic
range/precision/resolution of the parameters for the quantum annealer, and also to re-iterate the possibilities of different input
formulation (either with different parameter values or through some
NP-complete reductions as discussed in \cite{Diff2}) for the same problem, which may  also pose a challenge to
the benchmarking task.

The driver Hamiltonians we study in this paper are known as  {\em
  stoquastic} Hamiltonians (we use the property that their ground state is
the uniform superposition state with all positive amplitudes so that
our LENS idea applies).
There are arguments that the stoquastic
Hamiltonian system can be efficiently simulated by quantum Monte Carlo (QMC),
see \cite{Non-stoquastic,Layla,obstruction1,obstruction2} and
references therein. On the other hand, there are
examples \cite{Amin-Choi,AKR} that show that the quantum annealing algorithm with the transverse-field
driver Hamiltonian and a certain formulation of the problem 
Hamiltonian\footnote{See \cite{PNAS-correction} for a correction to \cite{AKR}.} will take exponential time because of the presence of 
anti-crossings. While one may eliminate the perturbative crossings
\cite{Dickson-Amin,Dickson}, the question of eliminating the more
general anti-crossings remains open. 
Perhaps the advantage of QA algorithms for the stoquastic Hamiltonian system over the
classical algorithms for the NP-hard optimization problem is questionable (e.g. because of the efficient
QMC simulation). Nevertheless, from the algorithmic point of view, it is still of
the value to better understand the working of the QA algorithm. For
example, 
by understanding the causes of the formation of
an anti-crossing, one may be able to come up with possible ways to
overcome them (such as what we show here by changing the parameter values),
and this will directly (through the QMC simulation) result in an efficient quantum-inspired
classical algorithm for these problem instances. 
Furthermore, it is possible to gain insight from the 
stoquastic Hamiltonian system and design the corresponding
non-stoquastic counterpart to overcome the anti-crossing problem as we
discuss in  \cite{driver2}, where we study both
stoquastic and non-stoquastic $XX$-driver Hamiltonians and different driver graphs
and show their effects on the quantum evolution of the QA algorithm.

Finally, we remark that there is a work of quantum speedup in
stoquastic adiabatic quantum computation (stoqAQC) in
\cite{Fujii}. The proposed stoqAQC model there requires non-standard
basis measurements and  
demands high-quality qubits, and quantum error
correction which is necessary for any quantum computing device, to achieve the speedup.

\section*{Acknowledgement}
The author is grateful to Tameem Albash for his generous sharing/update knowledge of the current
state of art the  QA algorithm, especially in regard to topics
including the  min-gap
estimation, anti-crossing and perturbative crossing, quantum
tunneling; and for his many
comments on this paper; and for the permission to use the loop-gadgets
example (including Figure \ref{fig:LG})  he constructed in this paper. 
Special thanks also go to Daniel Lidar, for the invitation to resume
the AQC research, and for his comments on this paper.
I would also like to thank Adrian Lupascu and his QEO group members
for the comments.
The author also gratefully acknowledges the use of the Sage
Mathematics Software~\cite{sage}
for the computation of the eigenvalues and plotting the graphs in this paper. 
The research is based upon work partially supported by the Office of
the Director of National Intelligence (ODNI), Intelligence Advanced
Research Projects Activity (IARPA), via the U.S. Army Research Office
contract W911NF-17-C-0050. The views and conclusions contained herein are
those of the authors and should not be interpreted as necessarily
representing the official policies or endorsements, either expressed or
implied, of the ODNI, IARPA, or the U.S. Government. The U.S. Government
is authorized to reproduce and distribute reprints for Governmental
purposes notwithstanding any copyright annotation thereon.

\end{document}